\documentclass[12pt]{amsart}
\usepackage{}

\usepackage{cancel}
\usepackage{amsmath}
\usepackage{amsfonts}
\usepackage{amssymb}
\usepackage[all]{xy}           

\usepackage{bbding}
\usepackage{txfonts}
\usepackage{amscd}

\usepackage[shortlabels]{enumitem}
\usepackage{ifpdf}
\ifpdf
  \usepackage[colorlinks,final,backref=page,hyperindex]{hyperref}
\else
  \usepackage[colorlinks,final,backref=page,hyperindex,hypertex]{hyperref}
\fi
\usepackage{tikz}
\usepackage[active]{srcltx}

\makeatletter

\newtheorem{thm}{Theorem}[section]
\newtheorem{lem}[thm]{Lemma}
\newtheorem{cor}[thm]{Corollary}
\newtheorem{pro}[thm]{Proposition}
\theoremstyle{definition}
\newtheorem{ex}[thm]{Example}
\newtheorem{rmk}[thm]{Remark}
\newtheorem{defi}[thm]{Definition}

\newcommand {\emptycomment}[1]{}

\newcommand{\lon }{\,\rightarrow\,}
\newcommand{\be }{\begin{equation}}
\newcommand{\ee }{\end{equation}}

\newcommand{\dm}{$\frkD$-map~}
\newcommand{\dms}{$\frkD$-maps~}
\newcommand{\ddm}{$\huaD$-map~}
\newcommand{\ddms}{$\huaD$-maps~}

\newcommand{\g}{\mathfrak g}
\newcommand{\h}{\mathfrak h}

\newcommand{\huaG}{\mathcal{G}}

\newcommand{\huaD}{\mathcal{D}}

\newcommand{\huaO}{{\mathcal{O}}}

\newcommand{\frkl}{\mathfrak l}

\newcommand{\frkD}{\mathfrak D}

\newcommand{\dM}{\mathrm{d}}

\newcommand{\Vir}{\mathrm{Vir}}

\newcommand{\Courant}[1]{\left\llbracket  #1\right\rrbracket }

\newcommand{\br}[1]{   [ \cdot,    \cdot  ]   }

\newcommand{\CE}{\mathsf{CE}}

\newcommand{\Hom}{\mathrm{Hom}}

\newcommand{\Der}{\mathrm{Der}}
\newcommand{\NR}{\mathrm{NR}}

\newcommand{\Inn}{\mathrm{Inn}}

\newcommand{\gl}{\mathfrak {gl}}

\newcommand{\ad}{\mathrm{ad}}

\newcommand{\K}{\mathbb{K}}

\begin{document}

\title[Deformation maps of quasi-twilled Lie algebras]{Deformation maps of quasi-twilled Lie algebras}

\author{Jun Jiang}
\address{Department of Mathematics, Jilin University, Changchun 130012, Jilin, China}
\email{junjiang@jlu.edu.cn}

\author{Yunhe Sheng}
\address{Department of Mathematics, Jilin University, Changchun 130012, Jilin, China}
\email{shengyh@jlu.edu.cn}

\author{Rong Tang}
\address{Department of Mathematics, Jilin University, Changchun 130012, Jilin, China}
\email{tangrong@jlu.edu.cn}


\begin{abstract}
In this paper, we provide a unified approach to study the cohomology theories and deformation theories of various types of operators in the category of Lie algebras, including modified $r$-matrices, crossed homomorphisms, derivations, homomorphisms, relative Rota-Baxter operators, twisted Rota-Baxter operators, Reynolds operators and deformation maps of matched pairs of Lie algebras. The main ingredients are quasi-twilled Lie algebras. We introduce two types of deformation maps of a quasi-twilled Lie algebra. Deformation maps of type I unify  modified $r$-matrices, crossed homomorphisms, derivations and homomorphisms between Lie algebras, while  deformation maps of type II unify relative Rota-Baxter operators, twisted Rota-Baxter operators, Reynolds operators and deformation maps of matched pairs of Lie algebras. We further give the controlling algebras and cohomologies of these two types of deformation maps, which not only recover the existing results for  crossed homomorphisms, derivations, homomorphisms, relative Rota-Baxter operators, twisted Rota-Baxter operators and Reynolds operators, but also leads to some new results which are unable to solve before, e.g. the controlling algebras and cohomologies  of modified $r$-matrices and deformation maps of matched pairs of Lie algebras.

\end{abstract}

\renewcommand{\thefootnote}{}
\footnotetext{2020 Mathematics Subject Classification.    17B38, 17B40,17B56}

\keywords{quasi-twilled Lie algebra, deformation map, relative Rota-Baxter operator, twisted Rota-Baxter operator, crossed homomorphism,  modified $r$-matrix, cohomology, deformation}

\maketitle

\tableofcontents

\allowdisplaybreaks


\section{Introduction}

A classical approach to study a mathematical structure is to associate to it invariants. Among these, cohomology theories occupy a  central position as they enable for example to control deformation  or extension problems. The concept of a formal deformation of an algebraic structure began with the seminal
work of Gerstenhaber~\cite{Ge0,Ge} for associative
algebras. Nijenhuis and Richardson   extended this study to Lie algebras
~\cite{NR,NR2}. Deformations of other
algebraic structures such as pre-Lie algebras have also been
developed~\cite{CL}. More generally, deformation theory
for algebras over quadratic operads was developed by Balavoine~\cite{Bal}. For more
general operads we refer the reader to \cite{KSo,LV,Ma}, and the references therein.
There is a well known slogan, often attributed to Deligne, Drinfeld and Kontsevich:  {\em every reasonable deformation theory is controlled by a differential graded (dg) Lie algebra, determined up to quasi-isomorphism}. This slogan has been made into a rigorous theorem by Lurie  and Pridham, cf. \cite{Lu,Pr}.

It is also meaningful to deform {\em maps} compatible with given algebraic structures. Recently, the cohomology and deformation theories of various operators were established with fruitful applications, e.g morphisms   \cite{Barmeier,Borisov,Das,Fre,Fregier,FZ},   derivations  \cite{TFS},  $\huaO$-operators (also called relative Rota-Baxter operators)   \cite{Das20,Das1,TBGS,Uc},   crossed homomorphisms   \cite{Das-ch,JS,PSTZ},  twisted Rota-Baxter operators and  Reynolds operators   \cite{Das0,Das-tRB}. The key step in most of the above studies is to construct the controlling algebra, namely an algebra whose Maurer-Cartan elements are the given structures, using the method of derived brackets \cite{Kosmann-Schwarzbach,Ma-0,Vo}. Then twisting the controlling algebra by a Maurer-Cartan element, one can obtain the algebra that governs deformations of the given operator, as well as the coboundary operator in the deformation complex.

Modified $r$-matrices, namely solutions of the modified Yang-Baxter equation, are important operators that have deep applications in mathematical physics, e.g. the Lax equation and the factorization problem  \cite{RS1, STS,STS2}.  Even though there are very fruitful results for various kinds of operators as aforementioned, but the controlling algebra for modified $r$-matrices  is still unknown. On the other hand, the notion of deformation maps of matched pairs of Lie algebras was introduced in \cite{AM,AM1} in the study of classifying compliments. It is also useful to develop  the cohomology and deformation theories for deformation maps of matched pairs of Lie algebras.

We propose a unified approach to study all aforementioned operators. On the one hand, we recover all the existing theories. On the other hand, we obtain some new results. We give the controlling algebra of modified $r$-matrices, which is a curved $L_\infty$-algebra, and establish the cohomology and deformation theories for deformation maps of matched pairs of Lie algebras.  We study all the operators under the general framework of quasi-twilled Lie algebras, which are generalization of quasi-Lie bialgebras \cite{KS1,KS-quasi}, and include direct sum of Lie algebras, semidirect products, action Lie algebras and matched pairs of Lie algebras as particular cases. We introduce two types of deformation maps of a quasi-twilled Lie algebra. Deformation maps of type I unify modified $r$-matrices, crossed homomorphisms, derivations and homomorphisms between Lie algebras, and  deformation maps of type II unify relative Rota-Baxter operators, twisted Rota-Baxter operators, Reynolds operators and deformation maps of matched pairs of Lie algebras. We further give the controlling algebras and cohomologies of these two types of deformation maps, and realize the above purposes.

Note that simultaneous  deformations of parts of aforementioned operators and algebras were studied in \cite{Barmeier,Das-sim,Fregier,LST,WZ}, and it would be helpful to develop a unified approach to study simultaneous deformations, which will be considered in the future. 

The paper is organized as follows. In Section \ref{sec:qtl}, we introduce the notion of a  quasi-twilled Lie algebra  and give various examples. In Section \ref{sec:dualdef}, we introduce the notion of a deformation map  of type I of a quasi-twilled Lie algebra, and give its controlling algebra and cohomology. In particular, we obtain the curved $L_\infty$-algebra, whose Maurer-Cartan elements are modified $r$-matrices. In Section \ref{sec:def}, we introduce the notion of a deformation map  of type II of a quasi-twilled Lie algebra, and give its controlling algebra and cohomology. As a byproduct, we give the controlling algebra and cohomology of a deformation map of a matched pair of Lie algebras.

\vspace{2mm}
\noindent
{\bf Acknowledgements. } This research is supported by NSFC (11922110,12371029).

\section{Quasi-twilled Lie algebras}\label{sec:qtl}

In this section, we introduce the notion of a quasi-twilled Lie algebra, and give various examples.

Let $\g$ be a vector space. Define the graded vector space
$\oplus_{n=0}^{+\infty} \Hom(\wedge^{n+1}\g,\g)$
with the degree of elements in $\Hom(\wedge^n\g,\g)$ being $n-1$. For $f\in \Hom(\wedge^m\g,\g), g\in \Hom(\wedge^n\g,\g)$, the  Nijenhuis-Richardson bracket $[\cdot,\cdot]_\NR$ is defined by
$$ [f,g]_{\NR}:=f\circ g- (-1)^{(m-1)(n-1)}g\circ f,$$
with $f\circ g\in \Hom(\wedge^{m+n-1}\g,\g)$ being defined by
\begin{equation}
(f\circ g)(x_1,\cdots,x_{m+n-1}):=\sum_{\sigma\in S(n,m-1)} (-1)^\sigma f(g(x_{\sigma(1)},\cdots,x_{\sigma(n)}),x_{\sigma(n+1)}, \cdots,x_{\sigma(m+n-1)}),
\label{eq:fgcirc}
\end{equation}
where the sum is over $(n,m-1)$-shuffles. Recall that a permutation $\tau\in S_n$ is called an  $(i,n-i)$-shuffle if $\tau(1)<\cdots <\tau(i)$ and $\tau(i+1)<\cdots <\tau(n)$.
Then $\big(\oplus_{n=0}^{+\infty} \Hom(\wedge^{n+1}\g, \g),[\cdot,\cdot]_{\NR}\big)$ is a   graded Lie algebra \cite{NR,NR2}. With this setup, a Lie algebra structure on $\g$ is precisely a  solution $\pi\in \Hom(\wedge^2\g, \g)$ of the Maurer-Cartan equation
$$ [\pi,\pi ]_{\NR}=0.$$

Let $\g$ and $\h$ be vector spaces. The elements in $\g$ are denoted by $x$ and the elements in $\h$ are denoted by $u$. For a multilinear map $\kappa: \wedge^{k}\g\otimes\wedge^{l}\h\lon \g,$ we define $\hat{\kappa}\in\Hom(\wedge^{k+l}(\g\oplus\h), \g\oplus\h)$ by
\begin{equation*}
\hat{\kappa}\Big((x_1, u_1),\cdots, (x_{k+l}, u_{k+l})\Big)=\sum_{\tau\in S(k,l)}(-1)^{\tau}\Big(\kappa(x_{\tau(1)},\cdots,x_{\tau(k)},u_{\tau(k+1)},\cdots, u_{\tau(k+l)}),0\Big).
\end{equation*}
Similarly, for  a multilinear map $\kappa: \wedge^{k}\g\otimes\wedge^{l}\h\lon \h,$ we define $\hat{\kappa}\in\Hom(\wedge^{k+l}(\g\oplus\h), \g\oplus\h)$ by
\begin{equation*}
\hat{\kappa}\Big((x_1, u_1),\cdots, (x_{k+l}, u_{k+l})\Big)=\sum_{\tau\in S(k,l)}(-1)^{\tau}\Big(0, \kappa(x_{\tau(1)},\cdots,x_{\tau(k)},u_{\tau(k+1)},\cdots,u_{\tau(k+l)})\Big).
\end{equation*}
The linear map $\hat{\kappa}$ is called a {\bf lift} of $\kappa$. We define $\g^{k,l}=\wedge^{k}\g\otimes\wedge^{l}\h$. Then  $\wedge^{n}(\g\oplus\h)\cong\oplus_{k+l=n}\g^{k,l}$ and $\Hom(\wedge^{n}(\g\oplus\h),\g\oplus\h)\cong(\oplus_{k+l=n}\Hom(\g^{k,l},\g))\oplus(\oplus_{k+l=n}\Hom(\g^{k,l},\h))$, where the isomorphism is given by the lift. See \cite{LST, Uc} for more details. In the sequel, we will omit the notation $\hat{\cdot}$.

\begin{defi}
Let $(\huaG, [\cdot, \cdot]_\huaG)$ be a Lie algebra with a decomposition into two subspaces $\huaG=\g\oplus\h$. The triple $(\huaG, \g, \h)$
is called a {\bf quasi-twilled Lie algebra} if $\h$ is a Lie subalgebra of $(\huaG, [\cdot, \cdot]_\huaG)$.
\end{defi}

Let $(\huaG, \g, \h)$ be a quasi-twilled Lie algebra. Denote the Lie bracket of $\huaG$ by $\Omega$. Then there exists
$\pi\in\Hom(\wedge^{2}\g, \g), ~\rho\in\Hom(\g\otimes\h, \h),~ \mu\in\Hom(\wedge^{2}\h, \h), ~\eta\in\Hom(\g\otimes\h, \g)$ and $\theta\in\Hom(\wedge^{2}\g, \h)$, such that
 \begin{equation}\label{eq:o}
\Omega=\pi+\rho+\mu+\eta+\theta.
\end{equation}
More precisely, for all $x, y\in\g,~ u, v\in\h$, we have
\begin{eqnarray*}
  \Omega((x,u),(y,v))=\big(\pi(x,y)+\eta(x,v)-\eta(y,u),\mu(u,v)+\rho(x,v)-\rho(y,u)+\theta(x,y)\big).
\end{eqnarray*}

In fact, $[\Omega, \Omega]_{\NR}=0$ if and only if the following equations hold
\begin{eqnarray}\label{111}
\left\{\begin{array}{rcl}
~~[\mu, \mu]_{\NR} &=&0, \\
~~[\pi, \pi]_{\NR}+2\eta\circ\theta &=&0, \\
~~2[\mu, \eta]_{\NR}+[\eta, \eta]_{\NR}&=&0, \\
~~[\rho, \mu]_{\NR}+\rho\circ\eta&=&0, \\
~~[\rho, \theta]_{\NR}+[\pi, \theta]_{\NR}&=&0, \\
~~[\pi, \eta]_{\NR}+\eta\circ\rho&=&0, \\
~~[\mu, \theta]_{\NR}+[\pi, \rho]_{\NR}+\theta\circ\eta+\frac{1}{2}[\rho, \rho]_{\NR}&=&0.
\end{array}\right.
\end{eqnarray}

\begin{pro}\label{etarep}
With the above notations, $\sigma:\h\lon\gl(\g)$ is a representation of the Lie algebra $(\h, \mu)$ on the vector space $\g$, where
$$
\sigma(v)x=-\eta(x, v), \quad \forall x\in\g, v\in\h.
$$
\end{pro}
\begin{proof}
By \eqref{111}, we have $2[\mu, \eta]_{\NR}+[\eta, \eta]_{\NR}=0$ which implies that $\sigma(v)x=-\eta(x, v)$ is a representation of the Lie algebra $(\h, \mu)$ on the vector space $\g$.
\end{proof}

Quasi-twilled Lie algebras can be viewed as natural generalizations of quasi-Lie bialgebras. Let $\g$ be a vector space and $\pi\in\wedge^2\g^*\otimes\g, ~\mu\in\g^*\otimes\wedge^2\g$ and $\theta\in\wedge^3\g^*$. Recall that the quadruple $(\g, \pi,\mu, \theta)$ is called a {\bf quasi-Lie bialgebra} \cite{B,KS1,KS-quasi} if
\begin{equation*}
\frac{1}{2}\{\pi, \pi\}+\{\mu, \theta\}=0, \quad \{\pi, \mu\}=0, \quad \{\mu, \mu\}=0, \quad \{\pi, \theta\}=0,
\end{equation*}
where $\{\cdot, \cdot\}$ is the  big bracket, which is the canonical Poisson bracket on $T^*[2]\g[1]$. For convenience and to be consistence with the notations for quasi-twilled Lie algebras, we will also view $\pi\in\wedge^2\g^*\otimes\g,~ \mu\in\g^*\otimes\wedge^2\g$ and $\theta\in\wedge^3\g^*$ as maps in $ \Hom(\wedge^{2}\g, \g),  \Hom(\wedge^{2}\g^*, \g^*) $ and $ \Hom(\wedge^{2}\g, \g^*)$ respectively. Then it is well known that the double $\g\oplus\g^*$ is a Lie algebra and it follows that $(\g\oplus\g^*,  \g, \g^*)$ is a quasi-twilled Lie algebra.

In the sequel, we give various examples of quasi-twilled Lie algebras.

Let $(\g, [\cdot,\cdot]_\g)$ be a Lie algebra.  For $\lambda\in\mathbb{K}$, define a bracket operation  $[\cdot, \cdot]_M$ on $ \g\oplus\g $ by
$$
[(x, u), (y, v)]_M=([x, v]_\g-[y, u]_\g, \lambda[x, y]_\g+[u, v]_\g), \quad \forall x, y, u, v\in\g.
$$
That is $\pi=\rho=0, \eta=\mu=[\cdot, \cdot]_\g$ and $\theta=\lambda[\cdot, \cdot]_\g$ in \eqref{eq:o}. Then we have that $[\cdot, \cdot]_M$ is a Lie algebra structure on $\g\oplus\g$. Denote this Lie algebra by $\g\oplus _M\g$.

\begin{ex}\label{ex:modir}
Let $(\g, [\cdot,\cdot]_\g)$ be a Lie algebra.  Then $(\g\oplus _M\g, \g, \g)$ is a quasi-twilled Lie algebra.
\end{ex}

Let $\rho:\g\to \Der(\h)$ be an action of a Lie algebra $\g$ on a Lie algebra $\h$. For any $\lambda\in \K$, then $(\g\oplus\h,[\cdot,\cdot]_\rho)$ is a Lie algebra, where the Lie bracket $[\cdot,\cdot]_\rho$ is given by
\begin{equation}
  [(x, u), (y, v)]_\rho=([x,y]_\g, \rho(x)v-\rho(y)u+\lambda[u,v]_\h),\quad \forall x,y\in\g, u,v \in\h.
\end{equation}
This Lie algebra is denoted by $\g\ltimes_\rho\h$ and called the {\bf action Lie algebra}.

\begin{ex}\label{ex:rRB1}
Let $\rho:\g\to \Der(\h)$ be an action of a Lie algebra $\g$ on a Lie algebra $\h$. Then $(\g\ltimes_\rho\h, \g, \h)$ is   a quasi-twilled Lie algebra.
\end{ex}

\begin{ex}\label{ex:rRB}
Let $\rho:\g\to \gl(V)$ be a representation of a Lie algebra $\g$ on a vector space $V$. Then $(\g\ltimes_\rho V, \g, V)$ is   a quasi-twilled Lie algebra, where $\g\ltimes_\rho V$ is the {\bf semidirect product Lie algebra} with the Lie bracket
$$
[(x, u), (y, v)]_{\rho}=([x, y]_\g, \rho(x)v-\rho(y)u), \quad \forall u, v\in V, x, y\in\g.
$$
\end{ex}

Let $(\g, [\cdot, \cdot]_\g)$ and $(\h, [\cdot, \cdot]_\h)$ be Lie algebras. Then there is the direct product Lie algebra  $(\g\oplus\h, [\cdot, \cdot]_{\oplus})$,  where
\begin{equation*}
[(x, u), (y, v)]_\oplus=([x, y]_\g, [u, v]_\h), \quad \forall x, y\in\g, u, v\in\h.
\end{equation*}

\begin{ex}\label{directpro}
The direct product Lie algebra $(\g\oplus\h,\g,\h)$ is a quasi-twilled Lie algebra.
\end{ex}

\begin{ex}\label{ex:twistedRB}
  Let $\rho:\g\to \gl(V)$ be a representation of a Lie algebra $\g$ on a vector space $V$ and $\omega\in\Hom(\wedge^2\g,V)$ be a $2$-cocycle. Then $(\g\oplus V, [\cdot,\cdot]_{\rho,\omega})$ is a Lie algebra, where the Lie bracket $[\cdot,\cdot]_{\rho,\omega}$ is given by
\begin{equation}
  [(x, u), (y, v)]_{\rho,\omega}=([x,y]_\g, \rho(x)v-\rho(y)u+\omega(x,y)),\quad \forall x,y\in\g, u,v \in V.
\end{equation}
Denote this Lie algebra by $\g\ltimes _{\rho,\omega}V$. Moreover, $(\g\ltimes _{\rho,\omega}V, \g, V)$ is   a quasi-twilled Lie algebra.
\end{ex}

\begin{ex}\label{ex:Reynolds}
  As a special case of Example \ref{ex:twistedRB}, consider $V=\g,~\rho=\ad$ and $\omega(x,y)=[x,y]_\g$. Then we obtain a quasi-twilled Lie algebra $(\g\ltimes_{\ad,\omega} \g,\g,\g)$.
\end{ex}

\begin{rmk}\label{rmk-nex}
 In fact, the above examples can be unified via extensions of Lie algebras. Recall that a Lie algebra  $\huaG$ is an extension of a Lie algebra $\g$ by a Lie algebra $\h$ if we have the following exact sequence:
\begin{equation}\label{seq:nonabelianext}
0\longrightarrow\h\stackrel{i}{\longrightarrow}\huaG\stackrel{\pi}{\longrightarrow} \g\longrightarrow 0.
\end{equation}
By  choosing a section $s:\g\lon\huaG$ of the extension \eqref{seq:nonabelianext},  $ \huaG$ is equal to  $ s(\g)\oplus i(\h) $, and $i(\h)$ is a subalgebra. Thus,   $(\huaG, s(\g), i(\h))$ is a quasi-twilled Lie algebra.
\end{rmk}

A {\bf matched pair of Lie algebras}  consists of a pair of Lie algebras  $(\g,\h)$, a  representation $\rho: \g\to\gl(\h)$ of $\g$ on $\h$ and a   representation $\eta: \h\to\gl(\g)$ of $\h$ on $\g$ such that
\begin{eqnarray}
\label{eq:mp1}\rho(x) [u,v]_{\h}&=&[\rho(x)u,v]_{\h}+[u,\rho(x) v]_{\h}+\rho\big((\eta(v)x\big)u-\rho\big(\eta(u)x\big) v,\\
\label{eq:mp2}\eta(u) [x, y]_{\g}&=&[\eta(u)x,y]_{\g}+[x,\eta(u) y]_{\g}+\eta\big(\rho(y) u\big)x-\eta\big(\rho(x)u\big)y,
\end{eqnarray}
 for all $x,y\in \g$ and $u,v\in \h$. We will denote a matched pair of Lie algebras by $(\g,\h;\rho,\eta)$, or simply by $(\g,\h)$.
Let $(\g,\h;\rho,\eta)$ be a matched pair of Lie algebras. Then there is a Lie algebra structure on the direct sum space $\g\oplus \h$ with the Lie bracket $[\cdot,\cdot]_{\bowtie }$ given by
\[[(x, u), (y, v)]_{\bowtie }=\big([x,y]_\g+\eta(u)y-\eta(v) x,[u,v]_\h+\rho(x)v-\rho(y)u\big).\]
Denote this Lie algebra by $\g\bowtie \h$.

\begin{ex}\label{ex:mch}
Let $(\g,\h;\rho,\eta)$ be a matched pair of Lie algebras. Then $(\g\bowtie \h, \g, \h)$ is a quasi-twilled Lie algebra. Note that the representation $\eta$ in the definition of a matched pair and the $\eta$ in \eqref{eq:o} are related as follows: $\rho(x)v=\rho(x,v),~\eta(u)(y)=-\eta(y,u)$.
\end{ex}

\begin{rmk}
 In Section \ref{sec:dualdef},  we will introduce the notion of a deformation map of  type I of a quasi-twilled Lie algebra, and will see that deformation maps of  type I of   quasi-twilled Lie algebras given in Examples \ref{ex:modir}, \ref{ex:rRB1}, \ref{ex:rRB} and \ref{directpro} are exactly modified $r$-matrices (solutions of the modified classical Yang-Baxter equation), crossed homomorphisms, derivations and Lie algebra homomorphisms.
\end{rmk}

\begin{rmk}
 In Section \ref{sec:def}, we will introduce the notion of a deformation map of  type II of a quasi-twilled Lie algebra, and  will see that deformation maps of  type II of   quasi-twilled Lie algebras given in Examples \ref{ex:rRB1}, \ref{ex:rRB}, \ref{ex:twistedRB} \ref{ex:Reynolds} and \ref{ex:mch} are exactly relative Rota-Baxter operators of weight $\lambda$, relative Rota-Baxter operators of weight $0$, twisted Rota-Baxter operators, Reynolds operators and deformation maps of a matched pair of Lie algebras. So deformation maps of  type II of   quasi-twilled Lie algebras provide a unified approach to study relative Rota-Baxter operators, twisted Rota-Baxter operators, Reynolds operators and deformation maps of a matched pair of Lie algebras.
\end{rmk}

\section{The controlling algebras and cohomologies of deformation maps of  type I}\label{sec:dualdef}

In this section,  $(\huaG, \g, \h)$ is always a quasi-twilled Lie algebra, and the Lie bracket on $\huaG$ is denoted by $$\Omega=\pi+\rho+\mu+\eta+\theta,$$
where $\pi\in\Hom(\wedge^{2}\g, \g),~ \rho\in\Hom(\g\otimes\h, \h),~ \mu\in\Hom(\wedge^{2}\h, \h),~ \eta\in\Hom(\g\otimes\h, \g)$ and $\theta\in\Hom(\wedge^{2}\g, \h)$.

\subsection{Deformation maps of  type I of a quasi-twilled Lie algebra}
In this subsection, we introduce the notion of deformation maps of  type I of a quasi-twilled Lie algebra, which unify modified $r$-matrices, crossed homomorphisms, derivations and homomorphisms between Lie algebras

\begin{defi}
Let $(\huaG, \g, \h)$ be a quasi-twilled Lie algebra. A {\bf deformation map of type I} ($\frkD$-map for short) of $(\huaG, \g, \h)$ is a linear map $D: \g\lon\h$   such that
\begin{equation*}
D\big(\eta(x, D(y))-\eta(y, D(x))+\pi(x, y)\big)=\mu(D(x), D(y))+\rho(x, D(y))-\rho(y, D(x))+\theta(x, y).
\end{equation*}
\end{defi}

\begin{rmk}
\dms may not exists. Consider the quasi-twilled Lie algebra $(\g\ltimes_{\rho,\omega} V, \g, V)$ given in Example \ref{ex:twistedRB} obtained from a representation $\rho$ of $\g$ on $V$ and a $2$-cocycle $\omega$. A linear map $D:\g\lon\h$  is a \dm   if and only if
\begin{equation*}
\omega(x, y)=-\Big(\rho(x)D(y)-\rho(y)D(x)-D([x, y]_\g)\Big)=\dM_{\CE}(-D)(x, y), \quad \forall x, y\in\g,
\end{equation*}
where $\dM_{\CE}$ is the corresponding Chevalley-Eilenberg coboundary operator of the Lie algebra $\g$ with coefficients in the representation $(V; \rho)$.
Thus, the quasi-twilled Lie algebra $(\g\ltimes_{\rho,\omega} V, \g, V)$ given in Example \ref{ex:twistedRB} admits  a \dm   if and only if $\omega$ is an exact $2$-cocycle.
\end{rmk}
Let $D:\g\lon\h$ be a linear map. Denote the graph of $D$ by $$\mathrm{Gr}(D)=\{(x, D(x))|x\in\g\}.$$
\begin{pro}\label{pro-mat}
A linear map $D:\g\lon\h$  is a \dm if and only if  $\mathrm{Gr}(D)$ is a subalgebra. In this case  $(\h,\mathrm{Gr}(D))$  is also a matched pair of Lie algebras.
\end{pro}
\begin{proof}
For all $(x, D(x)), (y, D(y))\in\mathrm{Gr}(D)$, we have
\begin{eqnarray*}
&&\Omega\Big((x, D(x)), (y, D(y))\Big)\\
&=&\Big(\pi(x, y)+\eta(x, D(y))-\eta(y, D(x)), \mu(D(x), D(y))+\rho(x, D(y))-\rho(y, D(x))+\theta(x, y)\Big).
\end{eqnarray*}
Thus, $\mathrm{Gr}(D)$ is a Lie subalgebra of $\huaG$, i.e. $\Omega\Big((x, D(x)), (y, D(y))\Big)\in\mathrm{Gr}(D)$, if and only if
$$
\mu(D(x), D(y))+\rho(x, D(y))-\rho(y, D(x))+\theta(x, y)=D\Big(\pi(x, y)+\eta(x, D(y))-\eta(y, D(x))\Big),
$$
namely $D$ is a \dm of the quasi-twilled Lie algebra $(\huaG, \g, \h)$.

Since $ \huaG=\mathrm{Gr}(D)\oplus\h$, it follows that  $(\h,\mathrm{Gr}(D))$  is  a matched pair of Lie algebras.
\end{proof}

\begin{rmk}
Let $(\huaG, \g, \h)$ be a quasi-twilled Lie algebra. Then we have $\huaG=\g\oplus\h$. Since any compliment of  $\h$ in $\huaG$ is isomorphic to a graph of a linear map $D$ from $\g$ to $\h$. Then by Proposition \ref{pro-mat}, to find a space $V$ which is a compliment of $\h$ in $\huaG$ such that  $(\h,V)$ is a matched pair of Lie algebras  is equivalent to find a \dm of the  quasi-twilled Lie algebra $(\huaG, \g, \h)$.
\end{rmk}

\begin{ex}
  Consider the quasi-twilled Lie algebra  $(\g\oplus_M \g, \g, \g)$ given in  Example \ref{ex:modir}.
In this case, a \dm of $(\g\oplus_M \g, \g, \g)$ is a linear map $D:\g\lon\g$ such that
$$
[D(x), D(y)]_\g-D([D(x), y]_\g+[x, D(y)]_\g)=-\lambda[x, y]_\g.
$$

Note that this equation is called the {modified Yang-Baxter equation} by Semenov-Tian-Shansky in the seminal work \cite{STS}, whose solutions are called    {\bf modified $r$-matrices} on the Lie algebra $(\g, [\cdot, \cdot]_\g)$.
\end{ex}

\begin{rmk}
Modified $r$-matrices   play an important role in
studying solutions of Lax equations \cite{RS1, STS,STS2}. Furthermore,  modified $r$-matrices are intimately related to particular factorization problems in
the corresponding Lie algebras and Lie groups. This factorization problem was considered by Reshetikhin and  Semenov-Tian-Shansky in the framework of the enveloping algebra of a Lie algebra with a modified $r$-matrix to study quantum integrable systems \cite{RS88}. Moreover, any modified $r$-matrix induces a post-Lie algebra  \cite{BGN}.
\end{rmk}

\begin{ex}
 Consider the quasi-twilled Lie algebra $(\g\ltimes_\rho \h, \g, \h)$ given in Example \ref{ex:rRB1} obtained from the action Lie algebra $\g\ltimes_\rho \h$.
In this case, a \dm of $(\g\ltimes_\rho\h, \g, \h)$ is a linear map $D:\g\lon\h$ such that
$$
D([x, y]_\g)=\rho(x)D(y)-\rho(y)D(x)+\lambda[D(x), D(y)]_\h,
$$
which   is exactly a {\bf crossed homomorphism of weight $\lambda$} from the Lie algebra  $\g$ to the Lie algebra $\h$ \cite{Lue}.
\end{ex}

\begin{rmk}
Note that crossed homomorphisms of weight $-1$ are   $\varepsilon$-derivations on the Lie algebras, which  play   crucial roles in the Jacobian
conjecture \cite{Zhao-1} and the  Mathieu-Zhao subspace theory \cite{VZ}. On the other hand,   crossed homomorphisms of weight $1$ are deeply related to the representation theory of  Cartan type Lie algebras \cite{PSTZ} and post-Lie algebras \cite{MQ}.
\end{rmk}

\begin{ex}
 Consider the quasi-twilled Lie algebra $(\g\ltimes_\rho V, \g, V)$ given in Example \ref{ex:rRB} obtained from the semidirect product Lie algebra $\g\ltimes_\rho V$.
In this case, a \dm   is a linear map  $D:\g\lon V$ such that
$$
D([x, y]_\g)=\rho(x) D(y) -\rho(y)D(x),
$$
which implies that $D$ is a {\bf derivation} from $(\g, [\cdot, \cdot]_\g)$ to  $V$. In particular, if $\rho$ is the adjoint representation of $\g$ on itself, then we obtain the usual derivation.
\end{ex}

\begin{ex}\label{ex:der}
 Consider the quasi-twilled Lie algebra $(\g\oplus \h, \g, \h)$ given in Example \ref{directpro} obtained from the direct product Lie algebra.
In this case, a \dm of $(\g\oplus\h, \g, \h)$ is a linear map $D:\g\lon\h$ such that
$$
D([x, y]_\g)=[D(x), D(y)]_\h,
$$
which   is exactly a {\bf Lie algebra homomorphism} from $(\g, [\cdot, \cdot]_\g)$ to $(\h, [\cdot, \cdot]_\h)$.
\end{ex}

At the end of this subsection, we illustrate the roles that \dms play in the twisting theory.
Let $D: \g\lon\h$ be a linear map. It follows that $D^{2}=0$ and $[\cdot, D]_{\NR}$ is a derivation of the graded Lie algebra $\big(\oplus_{n=0}^{+\infty } \Hom(\wedge^{n+1}(\g\oplus\h), \g\oplus\h),[\cdot,\cdot]_{\NR}\big)$. Then we gain  that $e^{[\cdot, D]_{\NR}}$ is an automorphism of the graded Lie algebra $\big(\oplus_{n=0}^{+\infty} \Hom(\wedge^{n+1}(\g\oplus\h), \g\oplus\h),[\cdot,\cdot]_{\NR}\big)$.

\begin{defi}
Let $D: \g\lon\h$ be a linear map. The transformation $\Omega^{D}\triangleq e^{[\cdot, D]_{\NR}}\Omega$ is called the {\bf twisting} of $\Omega$ by $D$.
\end{defi}

\begin{pro}\label{lemtwist}
With the above notations, we gain that
\begin{eqnarray}\label{twistor-1}
\Omega^D=e^{-D}\circ \Omega \circ (e^{D}\otimes e^{D})
\end{eqnarray}
 is a Lie algebra structure on $\g\oplus\h$ and $e^{D}:(\huaG,\Omega^{D})\lon(\huaG,\Omega)$
is an isomorphism between Lie algebras.
\end{pro}
\begin{proof}
Since $e^{[\cdot, D]_{\NR}}$ is an automorphism of the graded Lie algebra $\big(\oplus_{n=0}^{+\infty} \Hom(\wedge^{n+1}(\h\oplus\g), \h\oplus\g),[\cdot,\cdot]_{\NR}\big)$, we deduce that $\Omega^{D}$ is a Lie algebra structure on $\g\oplus\h$. Moreover,
by the similar computation in \cite{TS,Uc},  we obtain  that $\Omega^D=e^{-D}\circ \Omega \circ (e^{D}\otimes e^{D})$ and $e^{D}:(\huaG,\Omega^{D})\lon(\huaG,\Omega)$
is an isomorphism between Lie algebras.
\end{proof}

\begin{thm}\label{protwistd}
Let $(\huaG, \h, \g)$ be a quasi-twilled Lie algebra and $D:\g\to\h$ a linear map. Then  $((\huaG,\Omega^{D}),\g,\h)$ is a quasi-twilled Lie algebra. Moreover, write $\Omega^{D}=\pi^{D}+\rho^{D}+\mu^{D}+\eta^{D}+\theta^{D}$, $\pi^{D}\in\Hom(\wedge^{2}\g, \g),~ \rho^{D}\in\Hom(\g\otimes\h, \h),~ \mu^{D}\in\Hom(\wedge^{2}\h, \h),~ \eta^{D}\in\Hom(\g\otimes\h, \g)$ and $\theta^{D}\in\Hom(\wedge^{2}\g, \h)$.  We have
\begin{eqnarray*}
\pi^{D}(x, y)&=&\pi(x, y)+\eta(x, D(y))-\eta(y, D(x)),\\
\rho^{D}(x, v)&=&\rho(x, v)+\mu(D(x), v)-D(\eta(x, v)),\\
\mu^{D}(u,v)&=&\mu(u,v),\\
\eta^{D}(x,v)&=&\eta(x,v),\\
\theta^{D}(x,y)&=&\theta(x, y)+\rho(x, D(y))-\rho(y, D(x))-D(\pi(x, y))\\
&&+\mu(D(x), D(y))-D(\eta(x, D(y)))+D(\eta(y, D(x))),
\end{eqnarray*}
for all $x, y\in\g, u,v\in\h$.

Consequently, $D:\g\lon\h$ is a \dm if and only if the Lie algebras $(\g, \pi^D)$ and $(\h, \mu)$ form a matched pair of Lie algebras.
\end{thm}

\begin{proof}
By \eqref{twistor-1}, we have
\begin{equation*}
\Omega^{D}((0, u), (0, v))=e^{-D}\Omega(e^{D}(0, u), e^{D}(0, v))=\Omega((0, u), (0, v))\in\h.
\end{equation*}
Thus $\h$ is a Lie subalgebra of $(\huaG,\Omega^{D})$, which implies that  $((\huaG,\Omega^{D}),\g,\h)$ is a quasi-twilled Lie algebra.

For all $x, y\in\g, u, v\in\h$, by \eqref{twistor-1}, we have
\begin{eqnarray*}
&&(\pi^{D}(x, y), \theta^{D}(x, y))=\Omega^{D}\Big((x, 0), (y, 0)\Big)\\
&=&\Big(\pi(x, y)+\eta(x, D(y))-\eta(y, D(x)), \theta(x, y)+\rho(x, D(y))-\rho(y, D(x))-D(\pi(x, y))\\
&&+\mu(D(x), D(y))-D(\eta(x, D(y)))+D(\eta(y, D(x)))\Big),
\end{eqnarray*}
and
\begin{eqnarray*}
(\eta^{D}(x, v), \rho^{D}(x, v))&=&\Omega^{D}\Big((x, 0), (0, v)\Big)\\
&=&(\eta(x, v), \rho(x, v)+\mu(D(x), v)-D(\eta(x, v))),
\end{eqnarray*}
which completes the proof.
 \end{proof}

 Apply the above result to the quasi-Lie bialgebra $(\g, \pi,\mu, \theta)$, a \dm gives rise to a Lie bialgebra.

\begin{pro}
Let $D:\g\lon\g^*$ be a \dm of the quasi-twilled Lie algebra $(\g\oplus\g^*, \g, \g^*)$  obtained from the quasi-Lie bialgebra $(\g,\pi,\mu, \theta)$ such that $D=-D^*$. Then $((\g,\pi^D),(\g^*,\mu^D))$ is a Lie bialgebra.
\end{pro}
\begin{proof}
By Theorem \ref{protwistd},  $\theta^D=0$.
Thus $((\g,\pi^D),(\g^*,\mu^D))$ is a matched pair of Lie algebras.
Moreover, by $D=-D^*$, we can deduce that $((\g,\pi^D),(\g^*,\mu^D))$ is a Lie bialgebra.
\end{proof}

\subsection{The controlling algebra of \dms}

In this subsection, we give the controlling algebra of deformation maps of type I, which is  a curved $L_\infty$-algebra. An important byproduct is the controlling algebra of modified $r$-matrices, which is totally unknown before.

\begin{defi}\rm(\cite{KS})
Let $\g=\oplus_{k\in\mathbb Z}\g^k$ be a $\mathbb Z$-graded vector space. A {\bf  curved $L_\infty$-algebra} is a $\mathbb Z$-graded vector space $\g$ equipped with a collection $(k\ge 0)$ of linear maps $l_k:\otimes^k\g\lon\g$ of degree $1$ with the property that, for any homogeneous elements $x_1,\cdots, x_n\in \g$, we have
\begin{itemize}\item[\rm(i)]
{\em (graded symmetry)} for every $\sigma\in S_{n}$,
\begin{eqnarray*}
l_n(x_{\sigma(1)},\cdots,x_{\sigma(n)})=\varepsilon(\sigma)l_n(x_1,\cdots,x_n),
\end{eqnarray*}
\item[\rm(ii)] {\em (generalized Jacobi identity)} for all $n\ge 0$,
\begin{eqnarray*}\label{sh-Lie}
\sum_{i=0}^{n}\sum_{\sigma\in  S(i,n-i) }\varepsilon(\sigma)l_{n-i+1}(l_i(x_{\sigma(1)},\cdots,x_{\sigma(i)}),x_{\sigma(i+1)},\cdots,x_{\sigma(n)})=0,
\end{eqnarray*}

\end{itemize}where $\varepsilon(\sigma)=\varepsilon(\sigma;x_1,\cdots, x_n)$ is the   Koszul sign for a permutation $\sigma\in S_n$ and $x_1,\cdots, x_n\in \g$.
\end{defi}
We denote a curved $L_\infty$-algebra by $(\g,\{l_k\}_{k=0}^{+\infty})$. A curved $L_\infty$-algebra $(\g,\{l_k\}_{k=0}^{+\infty})$ with $l_0=0$ is exactly an $L_\infty$-algebra~\cite{LS}.

\begin{defi}
Let $(\g, \{l_k\}_{k=0}^{+\infty})$ be a curved $L_\infty$-algebra. A {\bf Maurer-Cartan element} is a degree $0$ element $x$ satisfying
\begin{equation*}
 l_0+\sum_{k=1}^{+\infty}\frac{1}{k!}l_k(x, \cdots, x)=0.
\end{equation*}
\end{defi}

Let $x$ be a Maurer-Cartan element of a curved $L_\infty$-algebra $(\g, \{l_k\}_{k=0}^{+\infty})$. Define $l_{k}^{x}:\otimes^{k} \g\lon \g ~~(k\geq1)$ by
\begin{equation*}
l_{k}^{x}(x_1,\cdots,x_k)=\sum_{n=0}^{+\infty}\frac{1}{n!}l_{k+n}(\underbrace{x,\cdots,x}_{n},x_1,\cdots, x_k).
\end{equation*}

\begin{thm}\label{twistLin}{\rm(\cite{DSV,Get})}
With the above notation,  $(\g,\{l_k^{x}\}_{k=1}^{+\infty})$ is an $L_\infty$-algebra which is called the twisted $L_\infty$-algebra by $x$.
\end{thm}

\begin{rmk}To ensure the convergence of the series appearing in the definition of Maurer-Cartan elements and Maurer-Cartan twistings above, one need the $L_\infty$-algebra being {filtered} given by Dolgushev and Rogers in \cite{Dolgushev-Rogers}, or weakly filtered given in \cite{LST}. Since all the 	$L_\infty$-algebras under consideration in the sequel satisfy the weakly filtered condition, so we will not mention this point anymore.
\end{rmk}

We recall Voronov's derived bracket construction \cite{Vo}, which is a powerful method for constructing a curved $L_\infty$-algebra.

\begin{defi}\rm(\cite{Vo})
A {\bf curved $V$-data} consists of a quadruple $(L,F,P,\Delta)$, where
\begin{itemize}
\item[$\bullet$] $(L=\oplus L^i,[\cdot,\cdot])$ is a graded Lie algebra,
\item[$\bullet$] $F$ is an abelian graded Lie subalgebra of $(L,[\cdot,\cdot])$,
\item[$\bullet$] $P:L\lon L$ is a projection, that is $P\circ P=P$, whose image is $F$ and kernel is a  graded Lie subalgebra of $(L,[\cdot,\cdot])$,
\item[$\bullet$] $\Delta$ is an element in $L^1$ such that $[\Delta,\Delta]=0$.
\end{itemize}
When $\Delta\in\ker(P)^{1}$ such that $[\Delta, \Delta]=0$, we refer to $(L,F,P,\Delta)$ as a {\bf $V$-data}.
\end{defi}

\begin{thm}\rm(\cite{Vo})\label{cV}
Let $(L,F,P,\Delta)$ be a curved $V$-data. Then  $(F,\{l_k\}_{k=0}^{+\infty})$ is a curved $L_\infty$-algebra, where $l_k$ are given by
\begin{equation*}
l_0=P(\Delta), \quad l_k(x_1, \cdots, x_n)=P([\cdots[[\Delta, x_1], x_2], \cdots, x_n]).
\end{equation*}
\end{thm}

Now we are ready to give the controlling algebra of \dms of a quasi-twilled Lie algebra.

\begin{thm}\label{VDLI}
Let $(\huaG, \g, \h)$ be a quasi-twilled Lie algebra. Then there is a curved $V$-data $(L,F,P,\Delta)$ as follows:
\begin{itemize}
\item[$\bullet$] the graded Lie algebra $(L,[\cdot,\cdot])$ is given by $(\oplus_{n=0}^{+\infty}\Hom(\wedge^{n+1}\g\oplus\h, \g\oplus\h), [\cdot, \cdot]_{\NR})$,
\item[$\bullet$] the abelian graded Lie subalgebra $F$ is given by $\oplus_{n=0}^{+\infty}\Hom(\wedge^{n+1}\g, \h)$,
\item[$\bullet$] $P:L\lon L$ is the projection onto the subspace $F$,
\item[$\bullet$] $\Delta=\pi+\rho+\mu+\eta+\theta$.
\end{itemize}
Consequently, we obtain a curved $L_{\infty}$-algebra $(\oplus_{n=0}^{+\infty}\Hom(\wedge^{n+1}\g, \h), l_0, l_1, l_2)$, where $l_0,l_1,l_2$ are given by
\begin{eqnarray*}
l_0&=&\theta\\
l_1(f)&=&[\pi+\rho, f]_{\NR}\\
l_2(f, g)&=&[[\mu+\eta, f]_{\NR}, g]_{\NR}.
\end{eqnarray*}

Furthermore, a linear map $D:\g\lon\h$ is a \dm of $(\huaG, \g, \h)$ if and only if $D$ is a Maurer-Cartan element of the above curved $L_{\infty}$-algebra.
\end{thm}
\begin{proof}
It is obvious  that $F$ is an abelian graded Lie subalgebra of $L$. Since $P$ is the projection onto $F$, it is obvious that $P^2=P$. Moreover, the kernel of $P$ is a graded Lie subalgebra of $L$. Thus $(L,F,P,\Delta)$ is a curved $V$-data. By Theorem \ref{cV},  we obtain a curved  $L_{\infty}$-algebra $(F, \{l_k\}_{k=1}^{+\infty})$, where $l_k$ are given by
\begin{equation*}
l_k(f_1, \cdots, f_n)=P([\cdots[[\Delta, f_1]_{\NR}, f_2]_{\NR}, \cdots, f_n]_{\NR}).
\end{equation*}

By Theorem \ref{cV}, we have $l_0=P(\Delta)=\theta$ and
\begin{eqnarray*}
l_1(f)&=&P([\pi+\rho+\mu+\eta+\theta, f]_{\NR})=[\pi+\rho, f]_{\NR},\\
l_2(f, g)&=&P([[\pi+\rho+\mu+\eta+\theta, f]_{\NR}, g]_{\NR})=[[\mu+\eta, f]_{\NR}, g]_{\NR},
\end{eqnarray*}
where $f\in\Hom(\wedge^{n}\g, \h), g\in\Hom(\wedge^{m}\g, \h)$. Since $F$ is   abelian and
$$
[[\pi+\rho+\mu+\eta+\theta, f]_{\NR}, g]_{\NR}\in\Hom(\wedge^{n+m}\g, \h),
$$
we have $l_k=0$ for all $k\geq 3$.
 Moreover, we have
\begin{eqnarray*}
  &&l_0(x, y)+l_1(D)(x, y)+\frac{1}{2}l_2(D, D)(x, y)\\
  &=&\theta(x, y)+[\pi+\rho, D]_{\NR}(x, y)+[[\mu+\eta, D]_{\NR}, D]_{\NR}(x, y)\\
  &=&\theta(x, y)+\rho(x, D(y))-\rho(y, D(x))-D(\pi(x, y))+\mu(D(x), D(y))-D(\eta(x, D(y)))+D(\eta(y, D(x))).
\end{eqnarray*}
Thus, $D$ is a Maurer-Cartan element of the curved $L_\infty$-algebra $(\oplus_{n=0}^{+\infty}\Hom(\wedge^{n+1}\g, \h), l_0, l_1, l_2)$ if and only if
$D$ is a \dm of $(\huaG, \g, \h)$. The proof is finished.
\end{proof}

As an immediate application of the above theorem, we obtain  the {\bf controlling algebra of modified $r$-matrices}, namely a curved $L_\infty$-algebra whose Maurer-Cartan elements are modified $r$-matrices.

\begin{cor}\label{cor:conR}
 Consider the quasi-twilled Lie algebra  $(\g\oplus_M \g, \g, \g)$ given in  Example \ref{ex:modir}.
Then $(\oplus_{n=0}^{+\infty}\Hom(\wedge^{n+1}\g, \g), l_0, l_1, l_2)$ is a curved $L_{\infty}$-algebra, where $l_0, l_1$ and $l_2$ are given by $l_0=\lambda[\cdot, \cdot]_\g, l_1=0$,
and
\emptycomment{
\begin{eqnarray*}
&&l_2(f, g)(x_1, \cdots, x_{p+q})\\
&=&(-1)^{p-1}\Big(\sum_{\sigma\in S(q,1,p-1)}-(-1)^{\sigma}f([x_{\sigma(q+1)}, g(x_{\sigma(1)},\cdots, x_{\sigma(q)})]_\g,x_{\sigma(q+2)},\cdots,x_{\sigma(p+q)})\\
&&-(-1)^{pq}\sum_{\sigma\in S(p,1,q-1)}-(-1)^{\sigma}g([x_{\sigma(p+1)}, f(x_{\sigma(1)},\cdots,x_{\sigma(p)})]_\g, x_{\sigma(p+2)},\cdots, x_{\sigma(p+q)}))\\
&&+(-1)^{pq}\sum_{\sigma\in S(p, q)}(-1)^{\sigma}[f(x_{\sigma(1)},\cdots,x_{\sigma(p)}), g(x_{\sigma(p+1)},\cdots,x_{\sigma(p+q)})]_\g
\end{eqnarray*}
}
\begin{eqnarray*}
&&l_2(f, g)(x_1, \cdots, x_{p+q})\\
&=&\sum_{\sigma\in S(q,1,p-1)}(-1)^{p}(-1)^{\sigma}f([x_{\sigma(q+1)}, g(x_{\sigma(1)},\cdots, x_{\sigma(q)})]_\g,x_{\sigma(q+2)},\cdots,x_{\sigma(p+q)})\\
&&-\sum_{\sigma\in S(p,1,q-1)}(-1)^{p(q+1)}(-1)^{\sigma}g([x_{\sigma(p+1)}, f(x_{\sigma(1)},\cdots,x_{\sigma(p)})]_\g, x_{\sigma(p+2)},\cdots, x_{\sigma(p+q)}))\\
&&+\sum_{\sigma\in S(p, q)}(-1)^{p(q+1)}(-1)^{\sigma}[f(x_{\sigma(1)},\cdots,x_{\sigma(p)}), g(x_{\sigma(p+1)},\cdots,x_{\sigma(p+q)})]_\g
\end{eqnarray*}

for all $f\in\Hom(\wedge^{p}\g, \g), g\in\Hom(\wedge^q\g, \g)$.

Moreover, Maurer-Cartan elements of this curved $L_\infty$-algebra are exactly modified $r$-matrices on the Lie algebra $(\g, [\cdot, \cdot]_\g)$.
\end{cor}

Through Theorem \ref{VDLI}, one can also recover the controlling algebras for crossed homomorphisms, derivations and Lie algebra homomorphisms after the suspension.  \emptycomment{Let $V =\oplus_{i\in\mathbb{Z}}V^i$ be a graded
vector space. Recall that  the suspension operator $s : V\lon sV$ is defined by assigning $V$ to the graded
vector space $s(V) =\oplus_{i\in\mathbb{Z}}(sV)^i$ with $(sV)^i=V^{i-1}$. There is a natural degree $1$ map $s : V\lon sV$
that is the identity map of the underlying vector space, sending $v\in V$ to its suspended copy $sv\in sV$.}

\begin{cor}\label{cor:conch}
Consider the quasi-twilled Lie algebra $(\g\ltimes_\rho \h, \g, \h)$ given in Example \ref{ex:rRB1} obtained from the action Lie algebra $\g\ltimes_\rho \h$.
Then $(\oplus_{n=1}^{+\infty}\Hom(\wedge^{n}\g, \h), \dM, \Courant{\cdot,\cdot})$  is a differential graded Lie algebra, where the differential $\dM$ is given by \begin{eqnarray}
\label{eq:dualdr}\dM(f)(x_1, \cdots, x_{n+1})
&=&\sum_{i=1}^{n+1}(-1)^{n+i}\rho(x_i)f(x_1,\cdots,\hat{x_i}, \cdots, x_{n+1})\\
\nonumber&&+\sum_{i<j}(-1)^{n+i+j-1}f([x_i, x_j]_\g, x_1\cdots, \hat{x_i}, \cdots, \hat{x_j}, \cdots, x_{n+1}),
\end{eqnarray}
for all $f\in\Hom(\wedge^{n}\g, \h)$,
 and the graded Lie bracket $\Courant{\cdot,\cdot}$ is given by
 \begin{eqnarray}\label{eq:duald}
\qquad \Courant{f, g}(x_1, \cdots, x_{p+q})&=&\sum_{\sigma\in S(p, q)}(-1)^{pq+1}(-1)^{\sigma}\lambda[f(x_{\sigma(1)},\cdots,x_{\sigma(p)}), g(x_{\sigma(p+1)},\cdots,x_{\sigma(p+q)})]_\h,
\end{eqnarray}
for all $f \in\Hom(\wedge^{p}\g, \h), g\in\Hom(\wedge^{q}\g, \h)$.  This differential graded Lie algebra is exactly the {\bf controlling algebra for crossed homomorphisms of weight $\lambda$} given in \cite{PSTZ}.
\end{cor}

\begin{cor}\label{con-deri}
Consider the quasi-twilled Lie algebra $(\g\ltimes_\rho V, \g, V)$ given in Example \ref{ex:rRB} obtained from the semidirect product Lie algebra $\g\ltimes_\rho V$.
Then $(\oplus_{n=1}^{+\infty}\Hom(\wedge^{n}\g, V), \dM)$ is a cochain complex, where $\dM$ is given by \eqref{eq:dualdr}.
A linear map $D\in\Hom(\g,V)$ is a derivation if and only if $\dM(D)=0$. Therefore, the cochain complex $(\oplus_{n=1}^{+\infty}\Hom(\wedge^{n}\g, V), \dM)$ can be viewed as the {\bf controlling algebra for  derivations} from the Lie algebra $(\g, [\cdot, \cdot]_\g)$ to   $V$. See \cite{TFS} for more details.
\end{cor}

\begin{cor}\label{con-homo}
 Consider the quasi-twilled Lie algebra $(\g\oplus\h, \g, \h)$ given in Example \ref{directpro} obtained from the direct product Lie algebra $\g\oplus\h$. Then $(\oplus_{n=1}^{+\infty}\Hom(\wedge^{n}\g, \h), \dM, \Courant{\cdot,\cdot})$  is a differential graded Lie algebra, where the graded Lie bracket $\Courant{\cdot,\cdot}$ is given by \eqref{eq:duald} and the differential $\dM$ is given by
\begin{eqnarray*}
\dM(f)(x_1, \cdots, x_{p+1})&=&\sum_{i<j}(-1)^{p+i+j-1}f([x_i, x_j]_\g, x_1\cdots, \hat{x}_i, \cdots, \hat{x}_j, \cdots, x_{p+1}),
\end{eqnarray*}
 for all $f \in\Hom(\wedge^{p}\g, \h)$.  This differential graded Lie algebra is exactly the {\bf controlling algebra for Lie algebra homomorphisms}. See \cite{Das, Fre, FZ, NR1} for more details.
\end{cor}

Let $D:\g\lon\h$ is a \dm of $(\huaG, \g, \h)$. By Theorem \ref{VDLI}, we obtain that $D$ is a  Maurer-Cartan element of the curved $L_{\infty}$-algebra  $(\oplus_{n=0}^{+\infty}\Hom(\wedge^{n+1}\g, \h), l_0, l_1, l_2)$. By Theorem \ref{twistLin}, we have the twisted $L_\infty$-algebra structure on $\oplus_{n=0}^{+\infty}\Hom(\wedge^{n+1}\g, \h)$ as following:
\begin{eqnarray}
\label{twist-1}l_1^{D}(f)&=&l_1(f)+l_2(D,f),\\
\label{twist-2}l_2^{D}(f,g)&=&l_2(f,g),\\
l^D_k&=&0,\,\,\,\,k\ge3,
\end{eqnarray}
where $f\in\Hom(\wedge^{p}\g, \h), g\in\Hom(\wedge^q\g, \h)$. Now  we are ready to give the
  $L_\infty$-algebra that controls deformations of the \dm $D$.

\begin{thm}\label{thm:defd1}
Let $D:\g\lon\h$ be a \dm of $(\huaG, \g, \h)$. Then for a linear map $D':\g\rightarrow \h$, $D+D'$ is a \dm if and only if $D'$ is a Maurer-Cartan element of the twisted $L_\infty$-algebra $(\oplus_{n=0}^{+\infty}\Hom(\wedge^{n+1}\g, \h), l^D_1, l^D_2)$, that is $D'$ satisfies the Maurer-Cartan equation:
$$
l_1^{D}(D')+\frac{1}{2}l_2^{D}(D',D')=0.
$$
\end{thm}
\begin{proof}
By Theorem \ref{VDLI}, $D+D'$ is a \dm  if and only if
  $$l_0+l_1(D+D')+\frac{1}{2}l_2(D+D',D+D')=0.$$
Since $D$ is a Maurer-Cartan element of the curved $L_{\infty}$-algebra  $(\oplus_{n=0}^{+\infty}\Hom(\wedge^{n+1}\g, \h), l_0, l_1, l_2)$, we deduce that  the above condition is equivalent to
  $$l_1(D')+l_2(D,D')+\frac{1}{2}l_2(D',D')=0.$$
That is, $l_1^{D}(D')+\frac{1}{2}l_2^{D}(D',D')=0,$
which implies that $D'$ is a Maurer-Cartan element of the twisted $L_\infty$-algebra $(\oplus_{n=0}^{+\infty}\Hom(\wedge^{n+1}\g, \h), l^D_1, l^D_2)$.
\end{proof}

\emptycomment{
Let $D$ be a \dm  of a quasi-twilled Lie algebra $(\huaG, \g, \h)$. Define $l_{k}^{D}:\otimes^{k} F\lon F ~~(k\geq1)$ by
\begin{equation*}
l_{k}^{D}(f_1,\cdots,f_k)=\sum_{n=0}^{+\infty}\frac{1}{n!}l_{k+n}(\underbrace{D,\cdots,D}_{n},f_1,\cdots, f_k).
\end{equation*}
More precisely, we have
\begin{equation}\label{lemgez}
l_1^D(f)=l_1(f)+l_2(D, f),\quad l_2^D(f, g)=l_2(f, g).
\end{equation}
and $l_k^D=0$ for $k\geq 3$.

\begin{thm}\label{thm:defd1}
With above notations, $(F,\{l_k^{D}\}_{k=1}^{2})$ is an $L_\infty$-algebra.  Moreover,  $D+D'$ is a \dm for a linear map $D':\g\lon\h$ if and only if $D'$ is a Maurer-Cartan element of the $L_{\infty}$-algebra $(\oplus_{n=0}^{+\infty}\Hom(\wedge^{n+1}\g, \h), \{l_k^D\}_{k=1}^{2})$.
\end{thm}
\begin{proof}
By Theorem \ref{twistLin}, $(F,\{l_k^{D}\}_{k=1}^{2})$ is an $L_\infty$-algebra.

Let $D+D'$ is a \dm of a quasi-twilled Lie algebra $(\huaG, \g, \h)$. By Theorem \ref{VDLI}, we have
\begin{equation*}
\theta+l_1(D+D')+\frac{1}{2}l_2(D+D', D+D')=0,
\end{equation*}
which implies that
\begin{equation*}
l_1(D')+l_2(D, D')+\frac{1}{2}l_2(D', D')=0.
\end{equation*}
By \eqref{lemgez}, we obtain $$l_1^{D}(D')+\frac{1}{2}l_2^{D}(D', D')=0.$$ Thus $D'$ is a Maurer-Cartan element of the $L_{\infty}$-algebra $(\oplus_{n=0}^{+\infty}\Hom(\wedge^{n+1}\g, \h), \{l_k^D\}_{k=1}^{2})$.
\end{proof}
}

\begin{rmk}
  Apply Theorem \ref{thm:defd1} to Corollary \ref{cor:conR}, one obtains the differential graded Lie algebra governing deformations of a modified $r$-matrix. Note that this differential graded Lie algebra was first given in \cite{JS2} via a different approach.  Apply Theorem \ref{thm:defd1} to Corollaries \ref{cor:conch}-\ref{con-homo},
  one can also obtain the algebras governing deformations of    crossed homomorphisms, derivations  and Lie algebra homomorphisms. 
\end{rmk}

\emptycomment{
\begin{rmk}\label{rmk111}
Let $D$ be a \dm of a quasi-twilled Lie algebra $(\huaG,\g,\h)$. By Theorem \ref{protwistd}, $((\huaG,\Omega^D),\g,\h)$  is a quasi-twilled Lie algebra   and $\theta^D=0$. Denote by
$$
L=\oplus^{+\infty}_{n=0}\Hom(\wedge^{n+1}(\g\oplus\h), \g\oplus\h), \quad\text{and} \quad F=\oplus_{n=0}^{+\infty}\Hom(\wedge^{n+1}\g, \h).
$$
Then $(L, F, P, \Delta=\Omega^{D})$ is a $V$-data, where   $P: L\lon L$ is the projection onto the space $F$. By Theorem \ref{VDLI},   $(F, \{\frkl_k\}_{k=1}^{+\infty})$ is an $L_{\infty}$-algebra, where
\begin{equation*}
\frkl_k(f_1, \cdots, f_k)=P([\cdots[[\Delta, f_1]_\NR, f_2]_\NR, \cdots, f_k]_\NR).
\end{equation*}
By direct calculation, the $L_{\infty}$-algebra   is the same with the $L_{\infty}$-algebra $(F,\{l_k^{D}\}_{k=1}^{2})$ given in Theorem \ref{thm:defd1}.
\end{rmk}
}

\subsection{Cohomology of \dms}

In this subsection, we introduce a cohomology theory of \dms which will unify the cohomologies of various operators.

\begin{lem}
Let $(\huaG, \g, \h)$ be a quasi-twilled Lie algebra and $D:\g\lon\h$ be a \dm of $(\huaG, \g, \h)$. Then
$$\pi^{D}(x, y)=\pi(x, y)+\eta(x, D(y))-\eta(y, D(x))$$
is a Lie algebra structure on $\g$ and
$$\rho^{D}(x) v=\rho(x, v)+\mu(D(x), v)-D(\eta(x, v))$$
is a representation of $(\g, \pi^{D})$ on $\h$.
\end{lem}
\begin{proof}
Since $\big((\g, \pi^D),(\h, \mu);\rho^{D},\eta\big)$ is a matched pair of Lie algebras, we gain that $(\g,\pi^{D})$ is a Lie algebra and $(\h;\rho^{D})$ is a representation of $(\g,\pi^{D})$.
\end{proof}

Let $\dM_{\CE}^D: \Hom(\wedge^{k}\g,\h)\lon\Hom(\wedge^{k+1}\g,\h)$ be the corresponding Chevalley-Eilenberg coboundary operator of the Lie algebra $(\g, \pi^D)$ with coefficients in the representation $(\h, \rho^D)$. More precisely, for all $f\in\Hom(\wedge^{k}\g,\h)$ and $x_1,\cdots,x_{k+1}\in\g$, we have
\begin{eqnarray}\label{defidD!}
&&\dM_{\CE}^{D}f(x_1,\cdots,x_{k+1})\\
\nonumber&=&\sum_{i=1}^{k+1}(-1)^{i+1}\rho^D(x_i) f(x_1,\cdots,\hat{x_i},\cdots,x_{k+1})\\
\nonumber&&+\sum_{i<j}(-1)^{i+j}f(\pi^D(x_i, x_j),x_1,\cdots,\hat{x_i},\cdots,\hat{x_j},\cdots,x_{k+1})\\
\nonumber&=&\sum_{i=1}^{k+1}(-1)^{i+1}\mu(D(x_i), f(x_1,\cdots,\hat{x_i},\cdots,x_{k+1}))+\sum_{i=1}^{k+1}(-1)^{i+1}\rho(x_i, f(x_1,\cdots,\hat{x_i},\cdots,x_{k+1}))\\
\nonumber&&-\sum_{i=1}^{k+1}(-1)^{i+1}D(\eta(x_i, f(x_1,\cdots,\hat{x_i},\cdots,x_{k+1})))\\
\nonumber&&+\sum_{i<j}(-1)^{i+j}f(\pi(x_i, x_j),x_1,\cdots,\hat{x_i},\cdots,\hat{x_j},\cdots,x_{k+1})\\
\nonumber&&+\sum_{i<j}(-1)^{i+j}f(\eta(x_i, D(x_j))-\eta(x_j, D(x_i)),x_1,\cdots,\hat{x_i},\cdots,\hat{x_j},\cdots,x_{k+1}).
\end{eqnarray}

Now, we define the cohomology of a \dm $D:\g\lon\h$. Define the space of $0$-cochains $C^{0}(D)$ to be $0$ and define the space of $1$-cochains $C^{1}(D)$ to be $\h$. For $n\geq 2$, define the space of $n$-cochains $C^{n}(D)$ by $C^{n}(D)=\Hom(\wedge^{n-1}\g,\h)$.

\begin{defi}
Let $D:\g\lon\h$ be a \dm of a quasi-twilled Lie algebra $(\huaG, \g, \h)$. The cohomology of  the cochain complex   $(\oplus_{i=0}^{+\infty}C^{i}(D), \dM_{\CE}^{D})$ is defined to be the {\bf cohomology  for the \dm} $D$.
\end{defi}

Denote the set of $n$-cocycles by $Z^{n}(D)$, the set of $n$-coboundaries by $B^{n}(D)$ and the $n$-th cohomology group by
\begin{equation*}
H^{n}(D)=Z^{n}(D)/B^{n}(D), \quad n\geq 0.
\end{equation*}

It is obvious that $u\in\h$ is closed if and only if
\begin{equation*}
\mu(D(x), u)+\rho(x, u)-D(\eta(x, u))=0, \quad \forall x\in\g,
\end{equation*}
and $f\in\Hom(\g,\h)$ is closed if and only if
\begin{eqnarray*}
&&\mu(D(x), f(y))-\mu(D(y), f(x))+\rho(x, f(y))-\rho(y, f(x))+D(\eta(y, f(x)))-D(\eta(x, f(y)))\\
&=&f(\pi(x, y))+f(\eta(x, D(y))-\eta(y, D(x))), \quad \forall x,y\in\g.
\end{eqnarray*}

Here we provide another intrinsic interpretation of the above coboundary operator.

Let $D$ be a \dm of a quasi-twilled Lie algebra $(\huaG, \g, \h)$. Recall the twisted $L_{\infty}$-algebra  $(\oplus_{n=0}^{+\infty}\Hom(\wedge^{n+1}\g, \h), l^D_1, l^D_2)$ from Theorem \ref{thm:defd1}, we have

\begin{pro}\label{deficoofdefor11}
With the above notations, for any $f\in\Hom(\wedge^{k}\g, \h)$, one has $$l_1^D(f)=(-1)^{k-1}\dM_{\CE}^{D}f.$$
\end{pro}
\begin{proof}
For $f\in\Hom(\wedge^{k}\g, \h)$, by \eqref{twist-1} and Theorem \ref{VDLI}, we have
\begin{eqnarray*}
l_1^D&=&l_1(f)+l_2(D, f)\\
&=& [\pi+\rho, f]_{\NR}+[[\mu+\eta, D]_{\NR}, f]_{\NR}\\
&=&[\pi^D+\rho^D, f]_\NR\\
&=&(-1)^{k-1}\dM_{\CE}^{D}(f),
\end{eqnarray*}which finishes the proof.
\end{proof}

\begin{ex}
Consider the quasi-twilled Lie algebra  $(\g\oplus_M \g, \g, \g)$ given in  Example \ref{ex:modir}. Let $D:\g\lon \g$ be a modified $r$-matrix. Then $(\g, [\cdot, \cdot]_D)$ is a Lie algebra, where the Lie bracket $[\cdot, \cdot]_D$ is given by
\begin{eqnarray*}
[x, y]_D=[x, D(y)]_\g-[y, D(x)]_\g.
\end{eqnarray*}
Moreover, the Lie algebra $(\g,[\cdot,\cdot]_D)$ represents on the vector space $\g$ via $\rho^D:\g\to\gl(\g)$ given by
$$
\rho^D(x)u=[D(x), u]_\h-D([x, u]_\g),\quad \forall u\in \h, x\in\g.
$$
The corresponding Chevalley-Eilenberg cohomology is taken to be the {\bf cohomology for a modified $r$-matrix}  \cite{JS2}.
\end{ex}
\begin{ex}
Consider the quasi-twilled Lie algebra $(\g\ltimes_\rho \h, \g, \h)$ given in Example \ref{ex:rRB1} obtained from the action Lie algebra $\g\ltimes_\rho \h$. Let $D:\g\lon \h$ be a crossed homomorphism of weight $\lambda$. Then   the Lie algebra $(\g,[\cdot,\cdot]_\g)$ represents on the vector space $\h$ via $\rho^D:\g\to\gl(\h)$ given by
$$
\rho^D(x)u=\rho(x)u+\lambda[D(x), u]_\h,\quad \forall u\in \h, x\in\g.
$$
The corresponding Chevalley-Eilenberg cohomology is taken to be the {\bf cohomology for a crossed homomorphism of weight $\lambda$} \cite{JS, PSTZ}.
\end{ex}
\begin{ex}
Consider the quasi-twilled Lie algebra $(\g\ltimes_\rho V, \g, V)$ given in Example \ref{ex:rRB} obtained from the action Lie algebra $\g\ltimes_\rho V$. Let $D:\g\lon V$ be a derivation from Lie algebra $(\g, [\cdot, \cdot]_\g)$ to the $\g$-module $V$. Then the   Chevalley-Eilenberg cohomology of the Lie algebra $(\g, [\cdot, \cdot]_\g)$ with coefficients in the representation $(V;\rho)$ is taken to be the {\bf cohomology for a derivation} from the Lie algebra $(\g, [\cdot, \cdot]_\g)$ to the $\g$-module $V$  \cite{TFS}.
\end{ex}

\begin{ex}
Consider the quasi-twilled Lie algebra $(\g\oplus\h, \g, \h)$ given in Example \ref{directpro} obtained from the direct product Lie algebra $\g\oplus\h$. Let $D:\g\lon \h$ be a  Lie algebra homomorphism. Then  the Lie algebra $(\g,[\cdot,\cdot]_\g)$ represents on the vector space $\h$ via $\rho^D:\g\to\gl(\h)$ given by
$$
\rho^D(x)u=[D(x), u]_\h,\quad \forall u\in \h, x\in\g.
$$
The corresponding Chevalley-Eilenberg cohomology is taken to be the {\bf cohomology for a Lie algebra homomorphism}. See \cite{Das, Fre, FZ, NR1} for more details.
\end{ex}


\section{The controlling algebras and cohomologies of deformation maps of type II}\label{sec:def}

In this section,  $(\huaG, \g, \h)$ is always a quasi-twilled Lie algebra, and the Lie bracket on $\huaG$ is denoted by $$\Omega=\pi+\rho+\mu+\eta+\theta,$$
where $\pi\in\Hom(\wedge^{2}\g, \g),~ \rho\in\Hom(\g\otimes\h, \h),~ \mu\in\Hom(\wedge^{2}\h, \h),~ \eta\in\Hom(\g\otimes\h, \g)$ and $\theta\in\Hom(\wedge^{2}\g, \h)$.

\subsection{Deformation maps of type II of a quasi-twilled Lie algebra}

In this subsection, we introduce the notion of deformation maps of  type II of a quasi-twilled Lie algebra, which unify relative Rota-Baxter operators, twisted Rota-Baxter operators, Reynolds operators and deformation maps of matched pairs of Lie algebras.
\begin{defi}
Let $(\huaG, \g, \h)$ be a quasi-twilled Lie algebra. A {\bf deformation map of type II} ($\huaD$-map for short) of $(\huaG, \g, \h)$ is a linear map $B: \h\lon\g$  such that
\begin{equation*}
\pi(B(u), B(v))+\eta(B(u), v)-\eta(B(v), u)=B\big(\mu(u, v)+\rho(B(u), v)-\rho(B(v), u)+\theta(B(u), B(v))\big).
\end{equation*}
\end{defi}

These two types of deformation maps are related as follows.

\begin{pro}
Let $D:\g\to\h$ be an invertible linear map. Then $D$ is a  \dm of $(\huaG, \g, \h)$ if and only if $D^{-1}:\h\to\g$ is a $\huaD$-map of $(\huaG, \g, \h)$.
\end{pro}
\begin{proof}
  It is straightforward.
\end{proof}

\begin{ex}
 Consider the quasi-twilled Lie algebra $(\g\ltimes_\rho \h, \g, \h)$ given in Example \ref{ex:rRB1} obtained from the action Lie algebra $\g\ltimes_\rho \h$.
In this case, a  \ddm  of $(\g\ltimes_\rho\h, \g, \h)$ is a linear map $B:\h\lon\g$ such that
$$
[B(u), B(v)]_\g=B(\rho(B(u))v-\rho(B(v))u+\lambda[u, v]_\h), \quad \forall u, v\in\h,
$$
which is exactly  is a {\bf relative Rota-Baxter operator of weight $\lambda$}   on $\g$ with respect to the action $(\h;\rho)$ \cite{BGN}.
\end{ex}

\begin{ex}
 Consider the quasi-twilled Lie algebra $(\g\ltimes_\rho V, \g, V)$ given in Example \ref{ex:rRB} obtained from the semidirect product Lie algebra $\g\ltimes_\rho V$.
In this case, a \ddm   is a linear map $B:V\lon\g$ such that
$$
[B(u), B(v)]_\g=B(\rho(B(u))v-\rho(B(v))u), \quad \forall u, v\in V,
$$
which is exactly  is a {\bf relative Rota-Baxter operator of weight $0$} (also called an $\huaO$-operator) on $\g$ with respect to the representation $(V;\rho)$ \cite{Ku,STS}.
\end{ex}

\begin{ex}\label{ex:trb}
 Consider the quasi-twilled Lie algebra $(\g\ltimes_{\rho,\omega} V, \g, V)$ given in Example \ref{ex:twistedRB} obtained from a representation $\rho$ of $\g$ on $V$ and a $2$-cocycle $\omega$.
In this case, a  \ddm  of $(\g\ltimes_{\rho,\omega} V, \g, V)$ is a linear map $B:V\lon\g$ such that
$$
[B(u), B(v)]_\g=B(\rho(B(u))v-\rho(B(v))u+\omega(B(u), B(v))), \quad \forall u, v\in V,
$$
which implies that $B$ is a {\bf twisted Rota-Baxter operator} \cite{Das0}.
\end{ex}

\begin{ex}
 Consider the quasi-twilled Lie algebra $(\g\ltimes_{\ad,\omega} \g, \g, \g)$ given in Example \ref{ex:Reynolds}. In this case, a  \ddm  of $(\g\ltimes_{\ad,\omega} \g, \g, \g)$ is a linear map $B:\g\lon\g$ such that
 $$
[B(u), B(v)]_\g=B\big([B(u),v]_\g-[B(v),u]_\g+[B(u), B(v)]_\g\big), \quad \forall u, v\in \g,
$$
which implies that $B$ is a {\bf Reynolds operator} \cite{Das0}.
\end{ex}

\begin{ex}
Consider the quasi-twilled Lie algebra $(\g \bowtie \h, \g, \h)$ given in Example \ref{ex:mch} obtained from a matched pair of Lie algebras.
In this case, a  \ddm  of $(\g\bowtie\h, \g, \h)$ is a linear map $B:\h\lon\g$ such that
$$
[B(u), B(v)]_\g-\eta(v)B(u)+\eta(u)B(v)=B\big([u, v]_\h+\rho(B(u)) v-\rho(B(v)) u\big), \quad \forall u, v\in\h,
$$
which is exactly a {\bf  deformation map} of a matched pair of Lie algebras introduced in \cite{AM}.
\end{ex}

At the end of this subsection, we illustrate the roles that \ddms play in the twisting theory.

Let $B: \h\lon\g$ be a linear map. It follows that $B^{2}=0$ and $[\cdot, B]_{\NR}$ is a derivation of the graded Lie algebra $\big(\oplus_{n=0}^{+\infty } \Hom(\wedge^{n+1}(\g\oplus\h), \g\oplus\h),[\cdot,\cdot]_{\NR}\big)$. Then we gain  that $e^{[\cdot, B]_{\NR}}$ is an automorphism of the graded Lie algebra $\big(\oplus_{n=0}^{+\infty} \Hom(\wedge^{n+1}(\g\oplus\h), \g\oplus\h),[\cdot,\cdot]_{\NR}\big)$.

\begin{defi}
The transformation $\Omega^{B}\triangleq e^{[\cdot, B]_{\NR}}\Omega$ is called the {\bf twisting} of $\Omega$ by $B$.
\end{defi}

Parallel to Lemma \ref{lemtwist}, we have the following lemma.
\begin{lem}\label{lemtwistB}
With the above notations,  $\Omega^B=e^{-B}\circ \Omega \circ (e^{B}\otimes e^{B})$ is a Lie algebra structure on $\huaG$.
\end{lem}

Obviously,  $\Omega^{B}$ is decomposed into the unique six substructures $\pi^{B}\in\Hom(\wedge^{2}\g, \g),~\rho^{B}\in\Hom(\g\otimes\h, \h),~\mu^{B}\in\Hom(\wedge^{2}\h, \h),~\eta^{B}\in\Hom(\g\otimes\h, \g),~\theta^{B}\in\Hom(\wedge^{2}\g, \h),~\xi^B\in\Hom(\wedge^{2}\h, \g)$. We have the following result.

\begin{thm}\label{thm:twist}
Write $\Omega=\pi+\rho+\mu+\eta+\theta$ and $\Omega^{B}=\pi^{B}+\rho^{B}+\mu^{B}+\eta^{B}+\theta^{B}+\xi^B$. Then we have
\begin{eqnarray*}
\pi^{B}(x, y)&=&\pi(x, y)-B(\theta(x, y)),\\
\rho^{B}(x, v)&=&\rho(x, v)-\theta(B(v), x),\\
\mu^{B}(u, v)&=&\mu(u, v)+\rho(B(u), v)-\rho(B(v), u)+\theta(B(u), B(v)),\\
\eta^{B}(x, v)&=&\eta(x, v)-\pi(B(v), x)-B(\rho(x, v))+B(\theta(B(v), x)),\\
\theta^{B}(x,y)&=&\theta(x,y),\\
\xi^B(u,v)&=&\eta(B(u), v)-\eta(B(v), u)+\pi(B(u), B(v))\\
&&-B(\mu(u, v))-B(\rho(B(u), v)+B(\rho(B(v), u))+B(\theta(B(v), B(u))),
\end{eqnarray*}
for all $x, y\in\g, u, v\in\h$.

Consequently, $B: \h\lon\g$  is a \ddm if and only if $((\huaG, \Omega^{B}),\g,\h)$ is a quasi-twilled Lie algebra.
\end{thm}
\begin{proof}
It follows from a direct but tedious computation. We omit details.
\end{proof}

\subsection{The controlling algebra of  \ddms}

In this subsection, we give the controlling algebra of deformation maps of type II, which is  an $L_\infty$-algebra. An important byproduct is the controlling algebra of deformation maps of matched pairs of Lie algebras introduced in \cite{AM}.

\begin{thm}\label{VDLI1}
Let $(\huaG, \g, \h)$ be a quasi-twilled Lie algebra.  Then there is a $V$-data $(L,F,P,\Delta)$ as following:
\begin{itemize}
\item[$\bullet$] the graded Lie algebra $(L,[\cdot,\cdot])$ is given by $(\oplus_{n=0}^{+\infty}\Hom(\wedge^{n+1}\g\oplus\h, \g\oplus\h), [\cdot, \cdot]_{\NR})$,
\item[$\bullet$] the abelian graded Lie subalgebra $F$ is given by $\oplus_{n=0}^{+\infty}\Hom(\wedge^{n+1}\h, \g)$,
\item[$\bullet$] $P:L\lon L$ is the projection onto the subspace $F$,
\item[$\bullet$] $\Delta=\pi+\rho+\mu+\eta+\theta$.
\end{itemize}
Consequently, we obtain an $L_{\infty}$-algebra $(\oplus_{n=0}^{+\infty}\Hom(\wedge^{n+1}\h, \g), l_1,l_2,l_3)$, where $l_1,l_2,l_3$ are given by
\begin{eqnarray*}
l_1(f)&=&[\mu+\eta, f]_{\NR},\\
l_2(f_1, f_2)&=&[[\pi+\rho, f_1]_{\NR}, f_2]_{\NR},\\
l_3(f_1, f_2, f_3)&=&[[[\theta, f_1]_{\NR}, f_2]_{\NR}, f_3]_{\NR},\\
l_k&=&0, \quad k\geq 4.
\end{eqnarray*}

Furthermore, a linear map $B:\h\lon\g$ is a \ddm of $(\huaG, \g, \h)$ if and only if $B$ is a Maurer-Cartan element of the above $L_{\infty}$-algebra.
\end{thm}
\begin{proof}
It is obvious that $F$ is an abelian graded Lie subalgebra of $L$ and $[\Delta,\Delta]_{\NR}=0,\Delta\in\ker(P)^{1}$.
Since $P$ is the projection onto $F$, it is obvious that $P^2=P$. Thus $(L,F,P,\Delta)$ is a $V$-data.

By Theorem \ref{cV},  we obtain an $L_{\infty}$-algebra $(\oplus_{n=0}^{+\infty}\Hom(\wedge^{n+1}\h, \g), \{l_k\}_{k=1}^{+\infty})$, where $l_k$ are given by
\begin{equation*}
l_k(f_1, \cdots, f_n)=P([\cdots[[\Delta, f_1]_{\NR}, f_2]_{\NR}, \cdots, f_n]_{\NR}).
\end{equation*}
For $f\in\Hom(\wedge^{k}\h, \g)$, we have
\begin{eqnarray*}
l_1(f)=P([\pi+\rho+\mu+\eta+\theta, f]_{\NR})=[\mu+\eta, f]_{\NR}.
\end{eqnarray*}
Moreover, for $f_1\in\Hom(\wedge^{p}\h, \g), f_2\in\Hom(\wedge^{q}\h, \g)$ and $f_3\in\Hom(\wedge^{r}\h, \g)$, we obtain
\begin{eqnarray*}
l_2(f_1, f_2)=P([[\pi+\rho+\mu+\eta+\theta, f_1]_{\NR}, f_2]_{\NR})=[[\pi+\rho, f_1]_{\NR}, f_2]_{\NR}
\end{eqnarray*}
and
\begin{eqnarray*}
l_3(f_1, f_2, f_3)=P([[[\pi+\rho+\mu+\eta+\theta, f_1]_{\NR}, f_2]_{\NR}, f_3]_{\NR})=[[[\theta, f_1]_{\NR}, f_2]_{\NR}, f_3]_{\NR}.
\end{eqnarray*}
Since $F$ is abelian and
$
[[[\pi+\rho+\mu+\eta+\theta, f_1]_{\NR}, f_2]_{\NR}, f_3]_{\NR}\in\Hom(\wedge^{p+q+r-1}\h, \g),
$
we have $l_k=0$ for all $k\geq 4$.

It is straightforward to obtain
\begin{eqnarray*}
  &&l_1(B)(u, v)+\frac{1}{2}l_2(B, B)(u, v)+\frac{1}{6}l_3(B, B, B)(u, v)\\
  &=&[\mu+\eta, B]_{\NR} (u,v)+\frac{1}{2}[[\pi+\rho, B]_{\NR}, B]_{\NR}(u,v)+\frac{1}{6}[[[\theta, B]_{\NR}, B]_{\NR}, B]_{\NR}(u, v)\\
  &=& -B(\mu(u, v))+\eta(B(u), v)-\eta(B(v), u) + \pi(B(u), B(v))+B(\rho(B(v), u)-\rho(B(u), v)) \\
&&+ B(\theta(B(v), B(u))).
\end{eqnarray*}
Thus, $B\in\Hom(\h, \g)$ is a Maurer-Cartan element of $(\oplus_{n=0}^{+\infty}\Hom(\wedge^{n+1}\h, \g), l_1,l_2,l_3)$ if and only if $B:\h\lon\g$ is a \ddm of $(\huaG, \g, \h)$. The proof is finished.
\end{proof}

\begin{cor}\label{cor:carrb}
 Consider the quasi-twilled Lie algebra $(\g\ltimes_\rho \h, \g, \h)$ given in Example \ref{ex:rRB1} obtained from the action Lie algebra $\g\ltimes_\rho \h$.
 Then $(\oplus_{n=1}^{+\infty}\Hom(\wedge^{n}\h, \g), \dM, \Courant{\cdot,\cdot})$  is a differential graded Lie algebra, where the differential $\dM:\Hom(\wedge^{p}\h, \g)\to \Hom(\wedge^{p+1}\h, \g)$ is given by
\begin{eqnarray*}
\dM f(u_1, \cdots, u_{p+1})
&=&\sum_{i<j}(-1)^{p+i+j-1}\lambda f([u_i, u_j]_\h, u_1, \cdots, \hat{u_i}, \cdots, \hat{u_j}, \cdots, u_{p+1}),
\end{eqnarray*}
and $\Courant{\cdot, \cdot}$ is given by
\begin{eqnarray}
\label{eql2}&&\Courant{f_1, f_2}(u_1, \cdots, u_{p+q})\\
\nonumber&=&-\sum_{\sigma\in S(q, 1, p-1)}(-1)^{\sigma}f_1(\rho(f_2(u_{\sigma(1)}, \cdots, u_{\sigma(q)}))( u_{\sigma(q+1)}), u_{\sigma(q+2)}, \cdots, u_{\sigma(q+p)})\\
\nonumber&&+\sum_{\sigma\in S(p, 1, q-1)}(-1)^{pq}(-1)^{\sigma}f_2(\rho(f_1(u_{\sigma(1)}, \cdots, u_{\sigma(p)}))( u_{\sigma(p+1)}), u_{\sigma(p+2)}, \cdots, u_{\sigma(q+p)})\\
\nonumber&&-\sum_{\sigma\in S(p, q)}(-1)^{pq}(-1)^{\sigma}[f_1(u_{\sigma(1)}, \cdots, u_{\sigma(p)}), f_2(u_{\sigma(p+1)}, \cdots, u_{\sigma(p+q)})]_\g,
\end{eqnarray}
for all $f_1\in\Hom(\wedge^{p}\h, \g), f_2\in\Hom(\wedge^{q}\h, \g)$. This differential graded Lie algebra is exactly the {\bf controlling algebra for relative Rota-Baxter operators of weight $\lambda$} initially given in \cite[Corollary 2.17]{TBGS2}. 
\end{cor}

\begin{cor}
 Consider the quasi-twilled Lie algebra $(\g\ltimes_\rho V, \g, V)$ given in Example \ref{ex:rRB} obtained from the semidirect product Lie algebra $\g\ltimes_\rho V$.
Then $(\oplus_{n=1}^{+\infty}\Hom(\wedge^{n}V, \g),\Courant{\cdot,\cdot})$  is a graded Lie algebra, where the graded Lie bracket $\Courant{\cdot,\cdot}$ is given by \eqref{eql2}.
  This graded Lie algebra is exactly the {\bf controlling algebra for $\huaO$-operators} on $\g$ with respect to the representation $(V;\rho)$ given in \cite[Proposition 2.3]{TBGS}.
\end{cor}

\begin{cor}\label{cortwistrb}
 Consider the quasi-twilled Lie algebra $(\g\ltimes_{\rho,\omega} V, \g, V)$ given in Example \ref{ex:twistedRB} obtained from a representation $\rho$ of $\g$ on $V$ and a $2$-cocycle $\omega$.
Then $(\oplus_{n=0}^{+\infty}\Hom(\wedge^{n+1}V, \g), l_2, l_3)$  is an $L_{\infty}$-algebra, where $l_2$ and $l_3$ are given by
\emptycomment{
\begin{eqnarray*}
&&l_2(f_1, f_2)(u_1, \cdots, u_{p+q})\\
&=&(-1)^{p}\Big(\sum_{\sigma\in S(q, 1, p-1)}(-1)^{\sigma}f_1(\rho(f_2(u_{\sigma(1)}, \cdots, u_{\sigma(q)}))( u_{\sigma(q+1)}), u_{\sigma(q+2)}, \cdots, u_{\sigma(q+p)})\\
&&-(-1)^{pq}\sum_{\sigma\in S(p, 1, q-1)}(-1)^{\sigma}f_2(\rho(f_1(u_{\sigma(1)}, \cdots, u_{\sigma(p)}))( u_{\sigma(p+1)}), u_{\sigma(p+2)}, \cdots, u_{\sigma(q+p)})\\
&&+(-1)^{pq}\sum_{\sigma\in S(p, q)}(-1)^{\sigma}[f_1(u_{\sigma(1)}, \cdots, u_{\sigma(p)}), f_2(u_{\sigma(p+1)}, \cdots, u_{\sigma(p+q)})]_\g\Big),
\end{eqnarray*}
and
\begin{eqnarray*}
&&l_3(f_1, f_2, f_3)(u_1, \cdots, u_{p+q+r-1})\\
&=&(-1)^{p+q+qr}\sum_{\sigma\in S(q, r, p-1)}(-1)^{\sigma}f_1(\Phi(f_2(u_{\sigma(1)}, \cdots, u_{\sigma(q)}), f_3(u_{\sigma(q+1)}, \cdots, u_{\sigma(q+r)})),\\
&& u_{\sigma(q+r+1)}, \cdots, u_{\sigma(p+q+r-1)})\\
&&-(-1)^{pq+pr}\sum_{\sigma\in S(p, r, q-1)}(-1)^{\sigma}f_2(\Phi(f_1(u_{\sigma(1)}, \cdots, u_{\sigma(p)}), f_3(u_{\sigma(p+1)}, \cdots, u_{\sigma(p+r)})),\\
&& u_{\sigma(p+r+1)}, \cdots, u_{\sigma(p+q+r-1)})\\
&&+(-1)^{pq+pr+qr+q+r}\sum_{\sigma\in S(p, q, r-1)}(-1)^{\sigma}f_3(\Phi(f_1(u_{\sigma(1)}, \cdots, u_{\sigma(p)}), f_2(u_{\sigma(p+1)}, \cdots, u_{\sigma(p+q)})),\\
&& u_{\sigma(p+q+1)}, \cdots, u_{\sigma(p+q+r-1)}).
\end{eqnarray*}
}
\begin{eqnarray*}
&&l_2(f_1, f_2)(u_1, \cdots, u_{p+q})\\
&=&\sum_{\sigma\in S(q, 1, p-1)}(-1)^{p}(-1)^{\sigma}f_1(\rho(f_2(u_{\sigma(1)}, \cdots, u_{\sigma(q)}))( u_{\sigma(q+1)}), u_{\sigma(q+2)}, \cdots, u_{\sigma(q+p)})\\
&&-(-1)^{p(q+1)}\sum_{\sigma\in S(p, 1, q-1)}(-1)^{\sigma}f_2(\rho(f_1(u_{\sigma(1)}, \cdots, u_{\sigma(p)}))( u_{\sigma(p+1)}), u_{\sigma(p+2)}, \cdots, u_{\sigma(q+p)})\\
&&+(-1)^{p(q+1)}\sum_{\sigma\in S(p, q)}(-1)^{\sigma}[f_1(u_{\sigma(1)}, \cdots, u_{\sigma(p)}), f_2(u_{\sigma(p+1)}, \cdots, u_{\sigma(p+q)})]_\g,
\end{eqnarray*}
and
\begin{eqnarray*}
&&l_3(f_1, f_2, f_3)(u_1, \cdots, u_{p+q+r-1})\\
&=&\sum_{\sigma\in S(q, r, p-1)}(-1)^{p+q+qr}(-1)^{\sigma}f_1(\Phi(f_2(u_{\sigma(1)}, \cdots, u_{\sigma(q)}), f_3(u_{\sigma(q+1)}, \cdots, u_{\sigma(q+r)})),\\
&& u_{\sigma(q+r+1)}, \cdots, u_{\sigma(p+q+r-1)})\\
&&-\sum_{\sigma\in S(p, r, q-1)}(-1)^{pq+pr}(-1)^{\sigma}f_2(\Phi(f_1(u_{\sigma(1)}, \cdots, u_{\sigma(p)}), f_3(u_{\sigma(p+1)}, \cdots, u_{\sigma(p+r)})),\\
&& u_{\sigma(p+r+1)}, \cdots, u_{\sigma(p+q+r-1)})\\
&&+\sum_{\sigma\in S(p, q, r-1)}(-1)^{pq+pr+qr+q+r}(-1)^{\sigma}f_3(\Phi(f_1(u_{\sigma(1)}, \cdots, u_{\sigma(p)}), f_2(u_{\sigma(p+1)}, \cdots, u_{\sigma(p+q)})),\\
&& u_{\sigma(p+q+1)}, \cdots, u_{\sigma(p+q+r-1)}).
\end{eqnarray*}

for all $f_1\in\Hom(\wedge^{p}\h, \g), f_2\in\Hom(\wedge^{q}\h, \g)$ and $f_3\in\Hom(\wedge^{r}\h, \g)$. This $L_{\infty}$-algebra is exactly the {\bf controlling algebra for twisted Rota-Baxter operators} given in  \cite[Theorem 3.2]{Das0}.
\end{cor}

\begin{cor}\label{cor:caRo}
 Consider the quasi-twilled Lie algebra $(\g\ltimes_{\ad,\Phi} \g, \g, \g)$ given in Example \ref{ex:Reynolds}.
Parallel to Corollary \ref{cortwistrb}, $(\oplus_{n=0}^{+\infty}\Hom(\wedge^{n+1}\g, \g), l_2, l_3)$  is an $L_{\infty}$-algebra. This $L_{\infty}$-algebra is exactly the {\bf controlling algebra for Reynolds operators}. See \cite{Das0} for more details.
\end{cor}

Theorem \ref{VDLI1} can not only recover some known results, but also gives rise to some new results, e.g. it gives rise to the {\bf controlling algebra for deformation maps of a matched pair of Lie algebras}.

\begin{cor}\label{cor:d}
Consider the quasi-twilled Lie algebra $(\g \bowtie \h, \g, \h)$ given in Example \ref{ex:mch} obtained from a matched pair of Lie algebras.
Then $(\oplus_{n=1}^{+\infty}\Hom(\wedge^{n}\h, \g), \dM, \Courant{\cdot,\cdot})$  is a differential graded Lie algebra, where $\dM:\Hom(\wedge^{p}\h, \g)\to \Hom(\wedge^{p+1}\h, \g)$ is given by
\begin{eqnarray*}
\dM f(u_1, \cdots, u_{p+1})
&=&\sum_{i=1}^{p+1}(-1)^{p+i}\eta(u_i)f(u_1, \cdots, \hat{u_i}, \cdots, u_{p+1}))\\
&&+\sum_{i<j}(-1)^{p+i+j-1}f([u_i, u_j]_\h, u_1, \cdots, \hat{u_i}, \cdots, \hat{u_j}, \cdots, u_{p+1})
\end{eqnarray*}
and the graded Lie bracket $\Courant{\cdot,\cdot}$ is given by \eqref{eql2}. Maurer-Cartan elements of this differential graded Lie algebra are exactly   deformation maps of a matched pair of Lie algebras.
\end{cor}

At the end of this subsection, we give the $L_\infty$-algebra governing deformations of \ddms of a quasi-twilled Lie algebra.

Let $B$ be a \ddm of the quasi-twilled Lie algebra $(\huaG, \g, \h)$. By Theorem \ref{twistLin}, we have the twisted $L_\infty$-algebra structure on $\oplus_{n=0}^{+\infty}\Hom(\wedge^{n+1}\h, \g)$ as following:
\begin{eqnarray}
\label{twist-rb-1}l_1^{B}(f)&=&l_1(f)+l_2(B, f)+\frac{1}{2}l_3(B, B, f),\\
\label{twist-rb-2}l_2^{B}(f_1,f_2)&=&l_2(f_1, f_2)+l_3(B, f_1, f_2),\\
\label{twist-rb-3}l_3^{B}(f_1,f_2,f_3)&=&l_3(f_1, f_2, f_3),\\
l^B_k&=&0,\,\,\,\,k\ge4,
\end{eqnarray}

\begin{thm}\label{thm:defd2}
Let $B$ be a \ddm of the quasi-twilled Lie algebra $(\huaG, \g, \h)$.
Then for a linear map $B':\h\lon\g$, $B+B': \h\lon\g$ is a \ddm of the quasi-twilled Lie algebra $(\huaG, \g, \h)$  if and only if $B'$ is a Maurer-Cartan element of the twisted $L_\infty$-algebra $(\oplus_{n=0}^{+\infty}\Hom(\wedge^{n+1}\h, \g),l_1^{B},l_2^{B},l_3^{B})$.
\end{thm}
\begin{proof}

By Theorem \ref{VDLI1}, the linear map $B+B'$ is a \ddm if and only if
$$l_1(B+B')+\frac{1}{2!}l_{2}(B+B', B+B')+\frac{1}{3!}l_{3}(B+B',B+B', B+B')=0.$$
Moreover, by the fact $B$ is a $\huaD$-map, we gain that  the above condition is equivalent to
\emptycomment{
\begin{eqnarray*}
&&\sum_{k=1}^{3}\frac{1}{k!}l_{k}(B+B', \cdots, B+B')\\
&=&l_1(B)+l_1(B')+\frac{1}{2}l_2(B', B')+l_2(B, B')+\frac{1}{2}l_2(B, B)+\frac{1}{3!}l_3(B', B', B')\\
&&\frac{1}{2}l_3(B, B', B')+\frac{1}{2}l_3(B, B, B')+\frac{1}{3!}l_3(B, B, B)\\
&=&\sum_{k=1}^{3}\frac{1}{k!}l_{k}^{B}(B', \cdots, B').
\end{eqnarray*}
}
\begin{eqnarray*}
l_1(B')+l_2(B, B')+\frac{1}{2}l_3(B, B, B')+\frac{1}{2}l_2(B', B')+\frac{1}{2}l_3(B, B', B')+\frac{1}{3!}l_3(B', B', B')=0.
\end{eqnarray*}
Thus $B+B': \h\lon\g$ is a \ddm of the quasi-twilled Lie algebra $(\huaG, \g, \h)$ if and only if $B'$ is a Maurer-Cartan element of the twisted $L_\infty$-algebra $(\oplus_{n=0}^{+\infty}\Hom(\wedge^{n+1}\h, \g),l_1^{B},l_2^{B},l_3^{B})$.
\end{proof}

Apply the above theorem to Corollary \ref{cor:d}, we obtain the differential graded Lie algebra governing deformations of deformation maps of a matched pair of Lie algebras.

\begin{cor}
  Let $B:\h\to\g$ be a deformation map of a matched pair $(\g,\h;\rho,\eta)$ of Lie algebras. Then $\big(s\Big(\oplus_{n=1}^{+\infty}\Hom(\wedge^{n}\h, \g)\Big),\dM^B,\Courant{\cdot,\cdot}\big)$ is a differential graded Lie algebra,  where $\Courant{\cdot,\cdot}$ is given by \eqref{eql2}, and $\dM^B$ is given by
  \begin{eqnarray*}
    &&\dM^Bf(u_1, \cdots, u_{p+1})\\
    &=&\sum_{i=1}^{p+1}(-1)^{p+i}\eta(u_i)f(u_1, \cdots, \hat{u_i}, \cdots, u_{p+1})-\sum_{i<j}(-1)^{p+i+j}f([u_i, u_j]_\h, u_1, \cdots, \hat{u_i}, \cdots, \hat{u_j}, \cdots, u_{p+1})\\
    &&+\sum_{i=1}^{p+1}(-1)^{p+i}B(\rho(f(u_{1}, \cdots, \hat{u_p}, \cdots, u_{p+1}))u_{i})-\sum_{i=1}^{p+1}(-1)^{i+p}f(\rho(B(u_i)), u_{1}, \cdots, \hat{u_i}, \cdots, u_{p+1})\\
&&+\sum_{i=1}^{p+1}(-1)^{i+p}[B(u_{i}), f(u_{1}, \cdots, \hat{u_i}, \cdots, u_{p+1})]_\g.
  \end{eqnarray*}
  Moreover, for a linear map $B':\h\to\g$, $B+B'$ is a deformation map if and only if $B'$ is a Maurer-Cartan element of the differential graded Lie algebra  $\big(\oplus_{n=1}^{+\infty}\Hom(\wedge^{n}\h, \g),\dM^B,\Courant{\cdot,\cdot}\big)$.
\end{cor}

\begin{rmk}
  Apply Theorem \ref{thm:defd2} to Corollaries \ref{cor:carrb}-\ref{cor:caRo}, one can also obtain the algebras governing deformations of relative Rota-Baxter operators, twisted Rota-Baxter operators and Reynolds operators.  See \cite{Das0, Das1, TBGS} for more details.
\end{rmk}

\subsection{Cohomologies of \ddms}
In this subsection, we introduce a cohomology theory of a \ddm  using the Chevalley-Eilenberg cohomology of a Lie algebra, and illustrate that it unifies the cohomologies of relative Rota-Baxter operators, twisted Rota-Baxter operators and Reynolds operators. It also helps us to define a cohomology theory of a deformation map  of a matched pair of Lie algebras.

\begin{lem}
 Let $B: \h\lon\g$ be a \ddm of a quasi-twilled Lie algebra $(\huaG, \g, \h)$. Then
$$
\mu^{B}(u, v)=\mu(u, v)+\rho(B(u), v)-\rho(B(v), u)+\theta(B(u), B(v))
$$
is a Lie algebra structure on $\h$ and
$$
\sigma(v)x=-\eta^{B}(x, v)=-\eta(x, v)+\pi(B(v), x)+B(\rho(x, v))-B(\theta(B(v), x))
$$
is a representation of the Lie algebra $(\h, \mu^{B})$ on the vector space $\g$.
\end{lem}
\begin{proof}
It follows from Theorem \ref{thm:twist} and Proposition \ref{etarep} directly.
\end{proof}

Let $\dM_{\CE}^B: \Hom(\wedge^{k}\h,\g)\lon\Hom(\wedge^{k+1}\h,\g)$ be the corresponding Chevalley-Eilenberg coboundary operator of the Lie algebra $(\h, \mu^B)$ with coefficients in the representation $(\g, \sigma)$. More precisely, for all $f\in\Hom(\wedge^{k}\h,\g)$ and $u_1,\cdots,u_{k+1}\in\h$, we have
\begin{eqnarray}\label{defidD!}
&&\dM_{\CE}^{B}f(u_1,\cdots,u_{k+1})\\
\nonumber&=&\sum_{i=1}^{k+1}(-1)^{i+1}\sigma(u_i)f(u_1,\cdots,\hat{u_i},\cdots,u_{k+1})+\sum_{i<j}(-1)^{i+j}f(\mu^B(u_i, u_j), u_1,\cdots,\hat{u_i},\cdots,\hat{u_j},\cdots,u_{k+1})\\
\nonumber&=&\sum_{i=1}^{k+1}(-1)^{i}\eta(f(u_1,\cdots,\hat{u_i},\cdots,u_{k+1}), u_i)+\sum_{i=1}^{k+1}(-1)^{i+1}\pi(B(u_i), f(u_1,\cdots,\hat{u_i},\cdots,u_{k+1}))\\
\nonumber&&+\sum_{i=1}^{k+1}(-1)^{i+1}B(\rho(f(u_1,\cdots,\hat{u_i},\cdots,u_{k+1})), u_i)-\sum_{i=1}^{k+1}(-1)^{i+1}B(\theta(B(u_i), f(u_1,\cdots,\hat{u_i},\cdots,u_{k+1})))\\
\nonumber&&+\sum_{i<j}(-1)^{i+j}f(\mu(u_i, u_j),u_1,\cdots,\hat{u_i},\cdots,\hat{u_j},\cdots,u_{k+1})\\
\nonumber&&+\sum_{i<j}(-1)^{i+j}f(\theta(B(u_i), B(u_j)),u_1,\cdots,\hat{u_i},\cdots,\hat{u_j},\cdots,u_{k+1})\\
\nonumber&&+\sum_{i<j}(-1)^{i+j}f(\rho(B(u_i), u_j)-\rho(B(u_j), u_i),u_1,\cdots,\hat{u_i},\cdots,\hat{u_j},\cdots,u_{k+1}).
\end{eqnarray}

Now, we define the cohomology of a \ddm $B:\h\lon\g$. Define the space of $0$-cochains $C^{0}(B)$ to be $0$ and define the space of $1$-cochains $C^{1}(B)$ to be $\g$. For $n\geq 2$, define the space of $n$-cochains $C^{n}(B)$ by $C^{n}(B)=\Hom(\wedge^{n-1}\h,\g)$.

\begin{defi}\label{defi:cohdm}
Let $(\huaG, \h, \g)$ be a quasi-twilled Lie algebra and $B:\h\lon\g$ be a \ddm of $(\huaG, \h, \g)$. The cohomology of  the cochain complex   $(\oplus_{i=0}^{+\infty}C^{i}(B), \dM_{\CE}^{B})$ is defined to be the {\bf cohomology  for the \ddm} $B$.
\end{defi}

Denote the set of $n$-cocycles by $Z^{n}(B)$, the set of $n$-coboundaries by $B^{n}(B)$ and the $n$-th cohomology group by
\begin{equation*}
H^{n}(B)=Z^{n}(B)/B^{n}(B), \quad n\geq 0.
\end{equation*}

It is obvious that $x\in\g$ is closed if and only if
\begin{equation*}
-\eta(x, u)+\pi(B(u), x)+B(\rho(x, u))-B(\theta(B(u), x))=0, \quad \forall u\in\h,
\end{equation*}
and $f\in\Hom(\h,\g)$ is closed if and only if
\begin{eqnarray*}
&&-\eta(f(v), u)+\eta(f(u), v)+\pi(B(u), f(v))-\pi(B(v), f(u))+B(\rho(f(v), u))-B(\rho(f(u), v))\\
&=&B(\theta(B(u), f(v)))-B(\theta(B(v), f(u)))+f(\mu(u, v)+\theta(B(u), B(v)))+f(\rho(B(u), v)-\rho(B(v), u)),
\end{eqnarray*}
for all $u, v\in\h$.

Here we provide an intrinsic interpretation of the above coboundary operator.

Let $B:\h\lon\g$ be a \ddm of a quasi-twilled Lie algebra $(\huaG, \g, \h)$. The twisted $L_\infty$-algebra $(\oplus_{n=0}^{+\infty}\Hom(\wedge^{n+1}\h, \g),l_1^{B},l_2^{B},l_3^{B})$ controls deformations of the \ddm $B$. Parallel to  Proposition \ref{deficoofdefor11}, we have the following proposition.

\begin{pro}\label{deficoofdefor1}
With the above notations, for any $f\in\Hom(\wedge^{k}\h, \g)$, one has $$l_1^B(f)=(-1)^{k-1}\dM_{\CE}^{B}f.$$
\end{pro}

Definition \ref{defi:cohdm} also recover the existing cohomology theories of relative Rota-Baxter operators, twisted Rota-Baxter operators and Reynolds operators.

\begin{ex}
 Consider the quasi-twilled Lie algebra $(\g\ltimes_\rho \h, \g, \h)$ given in Example \ref{ex:rRB1} obtained from the action Lie algebra $\g\ltimes_\rho \h$. Let  $B:\h\lon\g$ be a  relative Rota-Baxter operator of weight $\lambda$ on $\g$ with respect to the action $(\h;\rho)$. Then $(\h,[\cdot,\cdot]_B)$ is a Lie algebra, where the Lie bracket $[\cdot,\cdot]_B$ is given by
$$
[u,v]_B=\lambda[u, v]_\h+\rho(B(u))v-\rho(B(v))u,\quad \forall u,v\in \h.
$$
Moreover, the Lie algebra $(\h, [\cdot,\cdot]_B)$ represents on the vector space $\g$ via $\sigma:\h\to\gl(\g)$ given by
$$
\sigma(v)x=[B(v), x]_\g+B(\rho(x)v),\quad \forall v\in \h, x\in\g.
$$
The corresponding Chevalley-Eilenberg cohomology is taken to be the {\bf cohomology for the relative Rota-Baxter operator of weight $\lambda$}. See \cite{Das1, JSZ} for more details.
\end{ex}
\begin{ex}
 Consider the quasi-twilled Lie algebra $(\g\ltimes_\rho V, \g, V)$ given in Example \ref{ex:rRB} obtained from the semidirect product Lie algebra $\g\ltimes_\rho V$.
Let  $B:V\lon\g$ be a relative Rota-Baxter operator of weight $0$ or an $\huaO$-operator  on $\g$ with respect to the representation $(V;\rho)$. Then $(V,[\cdot,\cdot]_B)$ is a Lie algebra, where the Lie bracket $[\cdot,\cdot]_B$ is given by
$$
[u,v]_B=\rho(B(u))v-\rho(B(v))u,\quad \forall u,v\in V.
$$
Moreover, the Lie algebra $(V,[\cdot,\cdot]_B)$ represents on the vector space $\g$ via $\sigma:V\to\gl(\g)$ given by
$$
\sigma(u)x=[B(u),x]_\g+B(\rho(x)u),\quad \forall u\in V, x\in\g.
$$
The corresponding Chevalley-Eilenberg cohomology is taken to be the {\bf cohomology for the $\huaO$-operator} $B$. See \cite{TBGS} for more details.
\end{ex}

\begin{ex}
 Consider the quasi-twilled Lie algebra $(\g\ltimes_{\rho,\omega} V, \g, V)$ given in Example \ref{ex:twistedRB} obtained from a representation $\rho$ of $\g$ on $V$ and a $2$-cocycle $\omega$. Let $B:\h\lon\g$ be a twisted Rota Baxter operator. Then $(V,[\cdot,\cdot]_B)$ is a Lie algebra, where the Lie bracket $[\cdot,\cdot]_B$ is given by
$$
[u,v]_B=\rho(B(u))v-\rho(B(v))u+\omega(B(u), B(v)),\quad \forall u,v\in V.
$$
Moreover, the Lie algebra $(V, [\cdot,\cdot]_B)$ represents on the vector space $\g$ via $\sigma:V\to\gl(\g)$ given by
$$
\sigma(v)x=[B(v), x]_\g+B(\rho(x)v)-B(\omega(B(v), x)),\quad \forall v\in V, x\in\g.
$$
The corresponding Chevalley-Eilenberg cohomology is taken to be the {\bf cohomology for the twisted Rota-Baxter operator} $B$. See \cite{Das0} for more details.
\end{ex}

\begin{ex}
 Consider the quasi-twilled Lie algebra $(\g\ltimes_{\ad,[\cdot, \cdot]_\g} \g, \g, \g)$ given in Example \ref{ex:Reynolds}. Let $B:\g\lon\g$ be a Reynolds operator. Then $(\g,[\cdot,\cdot]_B)$ is a Lie algebra, where the Lie bracket $[\cdot,\cdot]_B$ is given by
$$
[x, y]_B=[B(x), y]_\g-[B(y), x]_\g+[B(x), B(y)]_\g,\quad \forall x, y\in \g.
$$
Moreover, the Lie algebra $(\g, [\cdot,\cdot]_B)$ represents on the vector space $\g$ via $\sigma:\g\to\gl(\g)$ given by
$$
\sigma(x)y=[B(x), y]_\g+B([y, x]_\g)-B([B(x), y]_\g),\quad \forall x, y\in\g.
$$
The corresponding Chevalley-Eilenberg cohomology is taken to be the {\bf cohomology for the Reynolds operator} $B$. See \cite{Das0} for more details.
\end{ex}

Definition \ref{defi:cohdm} leads the following definition of cohomologies of deformation maps of a matched pair of Lie algebras.

Consider the quasi-twilled Lie algebra $(\g \bowtie \h, \g, \h)$ given in Example \ref{ex:mch} obtained from a matched pair of Lie algebras. Let $B:\h\lon\g$ be a deformation map of a matched pair of Lie algebras. Then $(\h,[\cdot,\cdot]_B)$ is a Lie algebra, where the Lie bracket $[\cdot,\cdot]_B$ is given by
$$
[u, v]_B=[u, v]_\h+\rho(B(u))v-\rho(B(v))u,\quad \forall u,v\in \h.
$$
Moreover, the Lie algebra $(\h, [\cdot,\cdot]_B)$ represents on the vector space $\g$ via $\sigma:\h\to\gl(\g)$ given by
$$
\sigma(v)x=\eta(v)x+[B(v), x]_\g+B(\rho(x)v),\quad \forall v\in \h, x\in\g.
$$
\begin{defi}\label{defi:cohdmm}
The corresponding Chevalley-Eilenberg cohomology of the Lie algebra $(\h, [\cdot,\cdot]_B)$ with coefficients in the representation $(\g,\sigma) $ is taken to be the {\bf cohomology for the deformation map $B$ of a matched pair $(\g,\h)$ of Lie algebras}.
\end{defi}

\begin{rmk}
In \cite{AM}, authors constructed the Lie algebra structure $[\cdot, \cdot]_B$ on $\h$ via another approach, namely    transfer the Lie algebra structure on $\mathrm{Gr}(B)$ to $\h$.
\end{rmk}

\begin{rmk}
  In \cite{Das0, Das1, JSZ,TBGS}, it was showed that one can use the established cohomology theory to classify infinitesimal deformations of relative Rota-Baxter operators, twisted Rota-Baxter operators and Reynolds operators. Similarly, one can also study infinitesimal deformations of a deformation map of a matched pair and show that they are classified by  the second cohomology group of a deformation map of a matched pair given in Definition \ref{defi:cohdmm}. We leave the details to readers.
\end{rmk}




 \end{document}